\documentclass[aps,pra,superscriptaddress,longbibliography,showpacs,twocolumn]{revtex4-2}

\newif\ifapp
\newcommand{\lsp}{\hspace{0.1em}}
\usepackage{makecell}
\usepackage{physics}
\usepackage{xfrac}
\usepackage{graphicx}
\usepackage{subfigure}
\usepackage{mathdots}
\usepackage{amsfonts,amssymb,amsmath}
\usepackage{mathrsfs}
\usepackage{bbm}
\usepackage[]{graphics,graphicx,epsfig}
\usepackage{amsthm}
\usepackage{csquotes}
\MakeOuterQuote{"}
\usepackage{epstopdf}
\usepackage{tikz}

\usepackage{paralist}
\usepackage{diagbox}
\usepackage[inline]{enumitem}

\usepackage{tabularx}
\usepackage{setspace}
\usepackage{mathtools}
\usepackage{dsfont}
\usepackage{hyperref}

\hypersetup{
    breaklinks=true,
    colorlinks=true,       
    linkcolor=blue,          
    citecolor=red,        
    filecolor=magenta,      
    urlcolor=blue,           
    runcolor=cyan
}

\def\identity{\leavevmode\hbox{\small1\kern-3.8pt\normalsize1}}

\newtheorem{theorem}{Theorem}
\newtheorem{lemma}{Lemma}
\newtheorem{proposition}{Proposition}

\newcommand{\Cl}{\mathrm{Cl}}
\newcommand{\Sym}{\mathrm{Sym}}
\newcommand{\bbS}{\mathbb{S}}
\newcommand{\bbW}{\mathbb{W}}
\newcommand{\bbone}{\mathbbm{1}}

\newcommand{\supp}{\operatorname{supp}}

\newcommand{\caA}{\mathcal{A}}
\newcommand{\caD}{\mathcal{D}}
\newcommand{\caG}{\mathcal{G}}
\newcommand{\caH}{\mathcal{H}}

\newcommand{\caL}{\mathcal{L}}
\newcommand{\caP}{\mathcal{P}}
\newcommand{\caQ}{\mathcal{Q}}
\newcommand{\caS}{\mathcal{S}}
\newcommand{\caT}{\mathcal{T}}
\newcommand{\caU}{\mathcal{U}}

\newcommand{\rmd}{\mathrm{d}}
\newcommand{\rme}{\operatorname{e}}
\newcommand{\rmi}{\mathrm{i}}

\newcommand{\bs}{\mathrm{bs}}
\newcommand{\rmA}{\mathrm{A}}
\newcommand{\rmB}{\mathrm{B}}
\newcommand{\rmC}{\mathrm{C}}

\newcommand{\rmT}{\mathrm{T}}
\newcommand{\rmU}{\mathrm{U}}
\newcommand{\sgn}{\mathrm{sgn}}

\newcommand{\scrA}{\mathscr{A}}
\newcommand{\scrB}{\mathscr{B}}
\newcommand{\scrE}{\mathscr{E}}
\newcommand{\scrK}{\mathscr{K}}
\newcommand{\scrM}{\mathscr{M}}

\newcommand{\tPhi}{\tilde{\Phi}}
\newcommand{\tPsi}{\tilde{\Psi}}
\newcommand{\tscrM}{\tilde{\scrM}}

\newcommand{\lref}[1]{Lemma~\ref{#1}}
\newcommand{\lsref}[1]{Lemmas~\ref{#1}}

\newcommand{\thref}[1]{Theorem~\ref{#1}}
\newcommand{\thsref}[1]{Theorems~\ref{#1}}
\newcommand{\Thref}[1]{Theorem~\ref{#1}}

\newcommand{\pref}[1]{Proposition~\ref{#1}}

\newcommand{\Pref}[1]{Proposition~\ref{#1}}

\newcommand{\eref}[1]{Eq.~\textup{(\ref{#1})}}
\newcommand{\eqsref}[2]{Eqs.~(\ref{#1}) and (\ref{#2})}
\newcommand{\Eref}[1]{Equation~\textup{(\ref{#1})}}

\newcommand{\tref}[1]{Table~\ref{#1}}

\newcommand{\sref}[1]{Sec.~\ref{#1}}
\newcommand{\Sref}[1]{Section~\ref{#1}}

\newcommand{\fref}[1]{Fig.~\ref{#1}}

\newcommand{\aref}[1]{Appendix~\ref{#1}}

\def\<{\langle}  
\def\>{\rangle}  

\newcommand{\rcite}[1]{Ref.~\cite{#1}}

\newcommand{\equad}{\,\hphantom{=}\,}

\begin{document}
	\title{Optimal estimation of three parallel spins with genuine and restricted collective measurements}

  \author{Changhao Yi}
    \affiliation{State Key Laboratory of Surface Physics, Department of Physics, and Center for Field Theory and Particle Physics,
Fudan University, Shanghai 200433, China}
    \affiliation{Institute for Nanoelectronic Devices and Quantum Computing, Fudan University, Shanghai 200433, China}
    \affiliation{Shanghai Research Center for Quantum Sciences, Shanghai 201315, China}
	
	\author{Kai Zhou} 
	\affiliation{Laboratory of Quantum Information, University of Science and Technology of China, Hefei 230026, China}
	\affiliation{CAS Center For Excellence in Quantum Information and Quantum Physics, University of Science and Technology of China, Hefei 230026, China}
 
	\author{Zhibo Hou} 
	\affiliation{Laboratory of Quantum Information, University of Science and Technology of China, Hefei 230026,  China}
	\affiliation{CAS Center For Excellence in Quantum Information and Quantum Physics, University of Science and Technology of China, Hefei 230026,  China}
 \affiliation{Hefei National Laboratory, Hefei 230088, China}

 	\author{Guo-Yong Xiang}
	\affiliation{Laboratory of Quantum Information, University of Science and Technology of China, Hefei 230026, China}
	\affiliation{CAS Center For Excellence in Quantum Information and Quantum Physics, University of Science and Technology of China, Hefei 230026, China}
 \affiliation{Hefei National Laboratory, Hefei 230088,  China}
 
    \author{Huangjun~Zhu}
    \email{zhuhuangjun@fudan.edu.cn}
    \affiliation{State Key Laboratory of Surface Physics, Department of Physics, and Center for Field Theory and Particle Physics,
Fudan University, Shanghai 200433, China}
    \affiliation{Institute for Nanoelectronic Devices and Quantum Computing, Fudan University, Shanghai 200433, China}
    \affiliation{Shanghai Research Center for Quantum Sciences, Shanghai 201315, China}
 \affiliation{Hefei National Laboratory,  Hefei 230088,  China}

\begin{abstract}
 Collective measurements on identical and independent quantum systems can offer advantages in information
 extraction compared with individual measurements. However, little is known about the distinction between restricted collective measurements and genuine collective measurements in the multipartite setting. In this work we
 establish a rigorous performance gap based on a simple and old estimation problem, the estimation of a random
 spin state given three parallel spins. Notably, we derive an analytical formula for the maximum estimation fidelity
 of biseparable measurements and clarify its fidelity gap from genuine collective measurements. Moreover, we
 clarify the structure of optimal biseparable measurements. It turns out that  the maximum estimation fidelity can be
 achieved by two- and one-copy measurements assisted by one-way communication in one direction, but not the
 other way. Our work reveals a rich landscape of multipartite nonclassicality in quantum measurements instead
 of quantum states and is expected to trigger further studies.
\end{abstract}
	
 \date{\today}
\maketitle

\section{Introduction}
Quantum measurements are the key to extracting information from quantum systems \cite{NielC10book} and play crucial roles in various tasks in quantum information processing, such as quantum state estimation, quantum metrology, quantum communication, and quantum computation. When two or more quantum systems are available, collective measurements on all  quantum systems together  may extract more information than individual measurements
\cite{PereW91,MassP95,GisiP99,Mass00,BagaBGM06S,Zhu12the,ZhuH18U}, even if there is no entanglement or correlation among these quantum systems. This intriguing phenomenon is a manifestation of nonclassicality in quantum measurements rather than quantum states. Moreover, collective measurements are quite useful in many practical applications, including quantum state estimation \cite{MassP95,BagaBGM06S,HaahHJW16,ODonW16,Zhu12the,ZhuH18U},  direction estimation \cite{GisiP99,Mass00}, multiparameter estimation \cite{VidrDGJ14,LuW21,ChenCY22}, shadow estimation \cite{Aaro18,GriePS24,LiuLYZ24,LiYZZ24},
 quantum state discrimination
\cite{PereW91,HiggDBP11,MartHSS21}, quantum learning
\cite{HuanKP21,ChenCHL22,AharCQ22,GrewIKL24},
entanglement detection and distillation \cite{Horo03,LindMP98,DehaVDV03}, and nonlocality distillation \cite{EftaWC23}. The power of collective measurements has been demonstrated in a number of experiments \cite{HouTSZ18,TangHSZ20,ZhouYYH23,ConlVMP23,ConlELA23,TianYHX24}.

Although collective measurements are advantageous for many applications, their realization in  experiments is quite challenging, especially for multicopy  collective measurements.  Actually, almost all experiments in this direction are restricted to two-copy collective measurements, and a genuine three-copy collective measurement was realized only very recently \cite{ZhouYYH23}. In view of this situation, it is natural to ask whether there is a fundamental gap between restricted collective measurements on limited copies of quantum states and genuine collective measurements, which represent
the ultimate limit. This  problem is of interest to both foundational studies and practical applications. 
Unfortunately,  little is known about this problem, although the counterpart for quantum states has been well studied \cite{HoroHHH09,GuhnT09,BrunCPS13}. Conceptually, the very basic definitions remain to be clarified.
Technically, it is substantially more difficult to analyze the performance of restricted collective measurements.

In this work we start to explore the rich territory of multipartite nonclassicality in quantum measurements by virtue of a simple and old estimation problem, the estimation of a random spin state given three parallel spins \cite{MassP95,DerkBE98,LatoPT98,Haya98,BrusM99,AcinLP00,BagaBM02,BagaMM05,HayaHH05}.
To set the stage, we first introduce rigorous definitions of biseparable measurements (which encompass all restricted collective measurements) and genuine collective measurements. Then, we derive an analytical formula for the maximum estimation fidelity of biseparable measurements, which clearly demonstrates a fidelity gap from genuine collective measurements. Moreover, we clarify the structure of optimal biseparable measurements and highlight the role of mutually unbiased bases. In addition, we determine the maximum estimation fidelity 
based on one- and two-copy collective measurements assisted by one-way communication. Surprisingly, such strategies can reach the maximum estimation fidelity of biseparable measurements if the communication direction  is chosen properly. By contrast,  the maximum estimation fidelity achievable is strictly lower if the communication direction is reversed. Our work reveals a strict hierarchy of multicopy collective measurements and a plethora of nonclassical phenomena rooted in quantum measurements, which merit further studies.

The rest  of this paper is organized as follows. In \sref{sec:definition} we begin with the formal definitions of biseparable measurements and genuine collective measurements. In \sref{sec:state_estimation}, we 
review an old estimation problem and 
the concept of estimation fidelity.
In \sref{sec:CollectiveOpt} we review an optimal measurement for estimating three parallel spins, which is genuinely collective.
In \sref{sec:RestrictedOpt}, we determine the maximum estimation fidelities
of biseparable measurements, $2+1$ adaptive measurements, and $1+2$ adaptive measurements, respectively, and construct optimal estimation strategies explicitly. \Sref{sec: conclusion} summarizes this paper.

\section{Separable  and collective measurements}
\label{sec:definition}

\subsection{Quantum states and quantum measurements}

Let $\caH$ be a given finite-dimensional Hilbert space; let  $\caL(\caH)$ and $\rmU(\caH)$ be the space of linear operators and  the group of unitary operators on $\caH$, respectively. Quantum states on $\caH$ are represented by positive (semidefinite) operators of trace 1. The set of all quantum states on $\caH$ is denoted by $\caD(\caH)$ henceforth. Quantum measurements on $\caH$ can be described by positive operator-valued measures (POVMs) when post-measurement quantum states are  irrelevant \cite{NielC10book}. Mathematically, a POVM is composed of 
a set of positive operators that sum up to the identity operator, which is denoted by $\bbone$ henceforth. If we perform the POVM $\scrM=\{M_j\}_j$ on the quantum state $\rho$, then the probability of obtaining  outcome $j$ is $\tr(\rho M_j)$ according to the Born rule. To avoid trivial exceptions, we assume that no POVM element is equal to the zero operator in the rest of this paper.

Given two POVMs $\scrA=\{A_j\}_j$ and $\scrB=\{B_k\}_k$ on $\caH$, $\scrA$ is a \textit{coarse graining} of $\scrB$ if it can be realized by performing $\scrB$ followed by data processing \cite{MartM90,Zhu22}. In other words, the POVM elements $A_j$ of $\scrA$ can be expressed as follows:
\begin{equation}
A_j = \sum_{B_k\in\scrB}\Lambda_{jk}B_k \quad \forall \, A_j\in\scrA,
\end{equation}
where $\Lambda$ is a stochastic matrix satisfying $\Lambda_{jk}\geq 0$ and $\sum_j \Lambda_{jk} = 1$. A convex combination of $\scrA$ and $\scrB$ is the disjoint union $w\scrA\sqcup(1-w)\scrB$ of  $w\scrA$ and $(1-w)\scrB$ with $0\leq w\leq 1$, where 
\begin{equation}
w\scrA:=\{wA_j\}_j,\quad (1-w)\scrB:=\{(1-w)B_k\}_k,
\end{equation}
and zero POVM elements can be deleted. Convex combinations of three or more POVMs can be defined in a similar way.

\subsection{Separable and collective measurements}

Now, we turn to quantum states and POVMs on a bipartite system shared by Alice and Bob, where the total Hilbert space  is a tensor product of the form $\caH_\rmT=\caH_\rmA\otimes \caH_\rmB$. A quantum state $\rho$ on $\caH_\rmT$ is a product state if it is a tensor product of two states on $\caH_\rmA$ and $\caH_\rmB$, respectively. The state $\rho$ is
separable if it can be expressed as a convex sum of product states; otherwise, it is entangled \cite{HoroHHH09,GuhnT09}. Note that a pure state on $\caH_\rmT$ is separable if and only if (iff) it is a product state. 
A positive operator on $\caH_\rmT$ is separable if it is proportional to a separable state. A POVM (and similarly for the corresponding measurement) on $\caH_\rmT$ is \emph{separable} if every POVM element is separable.

Let $\scrA=\{A_j\}_j$ and $\scrB=\{B_k\}_k$ be two POVMs on $\caH_\rmA$ and $\caH_\rmB$, respectively. The tensor product of $\scrA$ and $\scrB$ is defined as $\scrA\otimes \scrB:=\{A_j\otimes B_k\}_{j,k}$. Such product POVMs are prominent examples of separable POVMs, but there are more interesting examples. A POVM is  $\rmA \rightarrow \rmB$ one-way adaptive if it has the form $\{A'_j\otimes B_{jk}'\}_{j,k}$, where $\scrA'=\{A_j'\}_j$ is a POVM on $\caH_\rmA$ and $\scrB_j'=\{B_{jk}'\}_k$ for each $j$ is a POVM on $\caH_\rmB$. Such a POVM can be realized by first performing the POVM $\scrA'$ on $\caH_\rmA$ and then performing the POVM $\scrB_j'$ on $\caH_\rmB$ if the first POVM yields outcome~$j$. 

\subsection{Biseparable measurements and genuine collective measurements}

Next, we turn to an $N$-partite quantum system with $N\geq 2$ being a positive integer. Now the total Hilbert space $\caH_\rmT$ can be expressed as  a tensor product of $N$ Hilbert spaces,
\begin{equation}\label{eq:HTtensorDecom}
	\caH_\rmT = \caH_1\otimes \caH_2\otimes \cdots\otimes \caH_N, 
\end{equation}
where $\caH_i$ (for $i=1,2,\ldots, N$) is the Hilbert space of party $i$. Let  $[N]$ denote the set $\{1,2,\ldots,N\}$. Let $\caP=\{I_1, I_2, \ldots, I_m\}$ be a set of disjoint nonempty subsets of $[N]$; then $\caP$ is a \emph{partition} of $[N]$ if $m\geq2$ and $\cup_{k=1}^m I_k=[N]$. Note that the order of $I_k$ and the order of elements in $I_k$ are irrelevant. Alternatively, the partition $\caP$ can be written as $(I_1| I_2| \cdots |I_m)$. The partition $\caP$ is complete if each set $I_k$ contains only one  element, that is $|I_k|=1$ for $k=1,2,\ldots, m$. The partition $\caP$ is a bipartition if $\caP$ contains two elements, that is, $|\caP|=m=2$. When $N=3$ for example, $[N]$ has one complete partition and three bipartitions, namely, 
\begin{equation}
\begin{gathered}
	(\{1\}|\{2\}|\{3\}),\quad  (\{1,2\}|\{3\}),\\  (\{1,3\}|\{2\}),\quad   (\{2,3\}|\{1\}),
\end{gathered}  
\end{equation}
which can be abbreviated as follows if there is no danger of confusion:
\begin{align}
	(1|2|3),\quad  (12|3),\quad (13|2),\quad   (23|1).
\end{align}

Given any partition $\caP=(I_1| I_2| \cdots |I_m)$ of $[N]$, the Hilbert space $\caH_\rmT$ can be expressed as a tensor product as follows:
\begin{equation}
	\caH_\rmT = \bigotimes_{k=1}^{m}\caH_{I_k},\quad \caH_{I_k} : = \bigotimes_{i\in I_k}\caH_i. 
\end{equation}
Here we implicitly assume that all the single-partite Hilbert spaces in the expansion of $\bigotimes_{k=1}^{m}\caH_{I_k}$ are eventually ordered as in \eref{eq:HTtensorDecom}. A quantum state $\rho$ on $\caH$ is $\caP$ separable (or separable with respect to the partition $\caP$) if $\rho$ can be expressed as follows:
\begin{align}
	\rho=\sum_l p_l \rho_l^{I_1}\otimes \rho_l^{I_2}\otimes \cdots\otimes \rho_l^{I_m},
\end{align}
where $\rho_l^{I_k}\in \caD(\caH_{I_k})$  and $\{p_l\}_l$ forms a probability distribution. The state $\rho$ is (completely) separable if it is $\caP$ separable when  $\caP$ is the complete partition. The state $\rho$ is \emph{biseparable} if $\rho$ can be expressed as follows:
\begin{align}
	\rho=\sum_{\caP, \, |\caP|=2} p_\caP \rho_\caP,
\end{align}
where the summation runs over all bipartitions of $[N]$, $\rho_\caP$ is $\caP$ separable, and $\{p_\caP\}_\caP$ forms a probability distribution. By contrast, the state $\rho$ has \emph{genuine multipartite entanglement} if it is not biseparable \cite{GuhnT09}. 

A positive operator on $\caH_\rmT$ is $\caP$ separable (biseparable) if it is proportional to  a quantum state that is $\caP$ separable (biseparable). A POVM $\scrM = \{M_j\}_j$ on $\caH_\rmT$ is $\caP$ separable if every POVM element $M_j$ is  $\caP$ separable. The POVM $\scrM$ is \emph{biseparable} if it is a coarse graining of a POVM of the form 
\begin{align}
	\scrK=\bigsqcup_{\caP,\, |\caP|=2} p_\caP\scrK_\caP,
\end{align}
where  the disjoint union  runs over all bipartitions of $[N]$, the POVM $\scrK_\caP$ is $\caP$ separable, $p_\caP\scrK_\caP$ means an element-wise product, and $\{p_\caP\}_\caP$ forms a probability distribution.  By contrast, a POVM is  \emph{genuinely collective} if it is not biseparable. Note that all POVM elements of a biseparable POVM are biseparable, but a POVM composed of biseparable POVM elements is not necessarily biseparable; see \eref{eq:POVMcol} in \sref{sec:CollectiveOpt} for an example.

\subsection{Biseparable measurements and genuine collective measurements for a tripartite system}

To be concrete, here we focus on  a tripartite quantum system, which represents the simplest nontrivial setting that can manifest multipartite quantum correlations. We assume that the whole system is  shared by Alice, Bob, and Charlie and label the three subsystems by $\rmA$, $\rmB$, and $\rmC$, which are more suggestive than the numbers 1, 2, and 3. Accordingly, the total Hilbert space can be expressed as  $\caH_\rmT=\caH_\rmA\otimes \caH_\rmB\otimes \caH_\rmC$, and the three bipartitions can be expressed  as $(\rmA\rmB|\rmC)$, $(\rmA\rmC|\rmB)$, and $(\rmB\rmC|\rmA)$.

A quantum state $\rho$  on $\caH_\rmT$ is (completely) separable if it can be expressed as follows:
\begin{align}
	\rho=\sum_l p_l \rho_l^{\rmA}\otimes \rho_l^{\rmB}\otimes \rho_l^{\rmC},
\end{align}
where $\rho_l^{\rmA}\in \caD(\caH_\rmA)$, $\rho_l^{\rmB}\in \caD(\caH_\rmB)$,   $\rho_l^{\rmC}\in \caD(\caH_\rmC)$, 
and $\{p_l\}_l$ forms a probability distribution. The state $\rho$ is $(\rmA\rmB|\rmC)$ separable 
[or separable with respect to the bipartition $(\rmA\rmB|\rmC)$] if $\rho$ can be expressed as follows:
\begin{align}
	\rho=\sum_l p_l \rho_l^{\rmA\rmB}\otimes \rho_l^{\rmC},
\end{align}
where $\rho_l^{\rmA\rmB}\in \caD(\caH_\rmA\otimes \caH_\rmB)$,  $\rho_l^{\rmC}\in \caD(\caH_\rmC)$, 
and $\{p_l\}_l$ forms a probability distribution.  In other words, $\rho$ is separable if $\rmA\rmB$ is regarded as a whole. In a similar way, we can define $(\rmA\rmC|\rmB)$ separable states and $(\rmB\rmC|\rmA)$ separable states. The state $\rho$ is biseparable if $\rho$ can be expressed as follows:
\begin{align}
	\rho=p_3 \rho_{(\rmA\rmB|\rmC)}+p_2 \rho_{(\rmA\rmC|\rmB)}+p_1 \rho_{(\rmB\rmC|\rmA)},
\end{align}
where the three quantum states $\rho_{(\rmA\rmB|\rmC)}$, $\rho_{(\rmA\rmC|\rmB)}$, and $\rho_{(\rmB\rmC|\rmA)}$ are $(\rmA\rmB|\rmC)$ separable, $(\rmA\rmC|\rmB)$ separable, and $(\rmB\rmC|\rmA)$ separable, respectively, and $\{p_l\}_{l=1}^3$ forms a probability distribution. 

\begin{figure}[t]
	\centering
	\includegraphics[width = 0.45\textwidth]{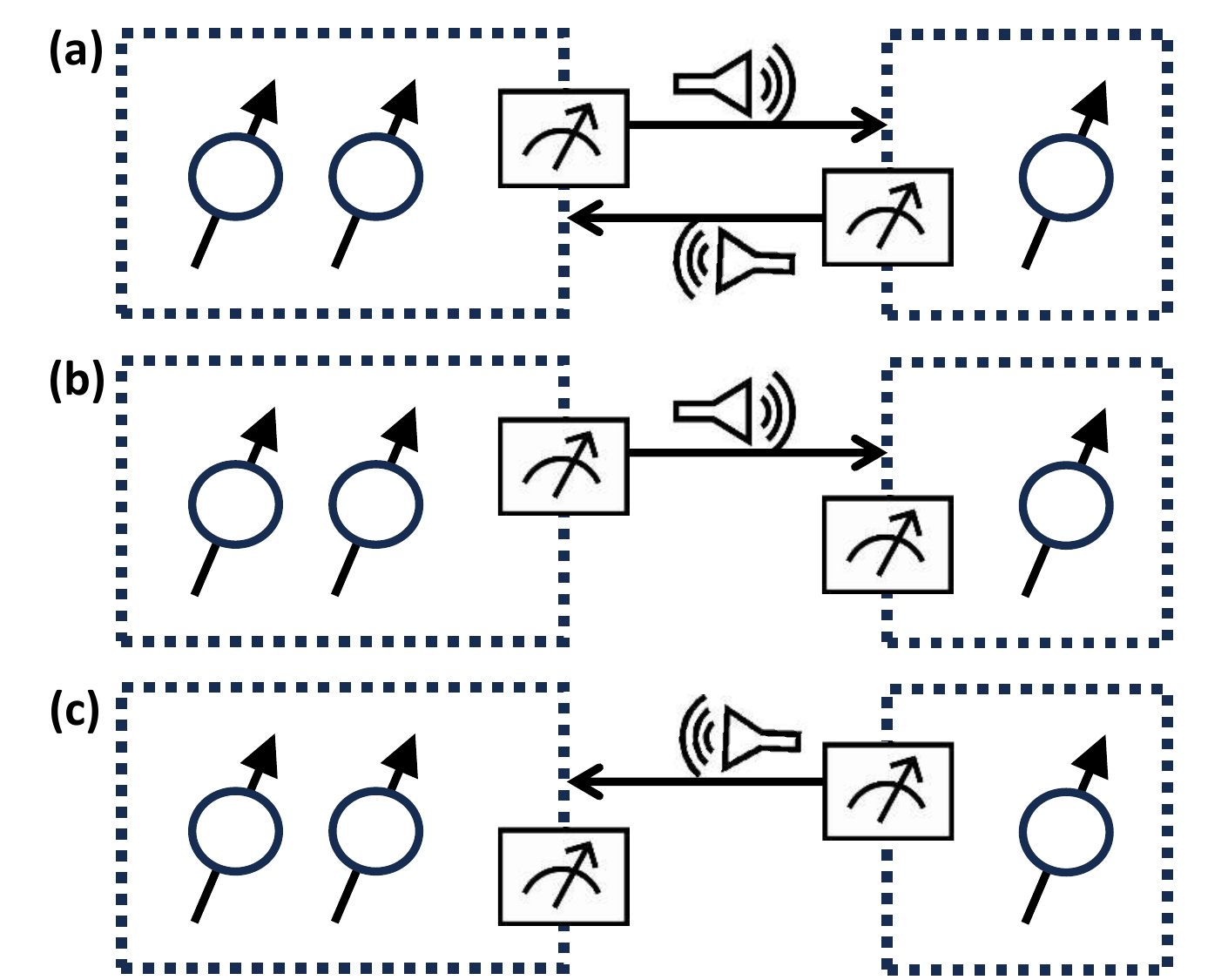}
	\caption{\label{fig:POVMRC}Three types of biseparable measurements on a tripartite quantum system: two-way adaptive (a), \mbox{$2+1$} adaptive (b), and $1+2$ adaptive (c).  Note that (a) contains (b) and (c) as  special cases.}
\end{figure}

Next, a positive operator on $\caH_\rmT$ is $(\rmA\rmB|\rmC)$ separable if it is proportional to a $(\rmA\rmB|\rmC)$ separable quantum state. A POVM $\scrM=\{M_j\}_j$ on $\caH_\rmT$ is  $(\rmA\rmB|\rmC)$ separable if every POVM element $M_j$ is $(\rmA\rmB|\rmC)$ separable. Generalization to the bipartitions $(\rmA\rmC|\rmB)$ and $(\rmB\rmC|\rmA)$ is immediate.  The POVM $\scrM$ is biseparable if it is a coarse graining of a POVM of the form 
\begin{align}
	\scrK= p_3\scrK_{(\rmA\rmB|\rmC)}\sqcup p_2\scrK_{(\rmA\rmC|\rmB)}\sqcup p_1 \scrK_{(\rmB\rmC|\rmA)},
\end{align}
where $p_1, p_2, p_3$ form a probability distribution and 
the  POVMs $\scrK_{(\rmA\rmB|\rmC)}$, $\scrK_{(\rmA\rmC|\rmB)}$, and $\scrK_{(\rmB\rmC|\rmA)}$ are $(\rmA\rmB|\rmC)$ separable, $(\rmA\rmC|\rmB)$ separable,  and $(\rmB\rmC|\rmA)$ separable, respectively.  A POVM is  genuinely collective if it is not biseparable as mentioned before.

Finally, we introduce three special types of biseparable POVMs  on $\caH_\rmT$. Given a bipartition, say $(\rmA\rmB|\rmC)$, a POVM is two-way adaptive if it can be realized by two-way communication between $\rmA\rmB$ and $\rmC$ (in addition to performing POVMs on $\rmA\rmB$ and $\rmC$, respectively); it is $2+1$ adaptive  if it can be realized by one-way communication from $\rmA\rmB$ to $\rmC$; it is $1+2$ adaptive if it can be realized by one-way communication from $\rmC$ to $\rmA\rmB$; see \fref{fig:POVMRC} for an illustration.

\section{Optimal quantum state estimation}
\label{sec:state_estimation}

\subsection{A simple estimation problem and estimation fidelity}
Here we reexamine an old estimation problem: A quantum device produces $N$ copies of a Haar-random pure state $\rho=|\psi\>\<\psi|$ on a $d$-dimensional Hilbert space $\caH$, and our task is to estimate the identity of $\rho$ based on quantum measurements 
\cite{MassP95,DerkBE98,LatoPT98,Haya98,BrusM99,GisiP99,Mass00,AcinLP00,BagaBM02,BagaMM05,HayaHH05}. The performance of an estimation protocol is quantified by the average fidelity. Suppose we perform a POVM $\scrM=\{M_j\}_j$ on $\rho^{\otimes N}$, then the probability of obtaining outcome $j$ reads $p_j = \tr(M_j \rho^{\otimes N})$. If we choose $\hat{\rho}_j$ as the estimator associated with outcome $j$, then the average estimation fidelity achieved by this protocol reads
\begin{equation}
\overline{F} = \sum_j \int_{\text{Haar}}d\psi \tr\bigl[(|\psi\>\<\psi|)^{\otimes N}M_j\bigr]\<\psi|\hat{\rho}_j|\psi\>,
\end{equation}
where the integral means taking the average over the ensemble of Haar-random pure states.

Let $\Sym_N(\caH)$ be the symmetric subspace in $\caH^{\otimes N}$ and $P_N$ the projector onto $\Sym_N(\caH)$.  Define
\begin{align}
\caQ(M_j) &:= (N+1)!\tr_{1,2,\ldots, N}[P_{N+1}(M_j\otimes \bbone)], \label{eq:Qmap}\\
F(\scrM)&:= \sum_j\frac{\|\caQ(M_j)\|}{d(d+1)\cdots (d+N)},\label{eq:EstimationFidDef}
\end{align}
where $\|\cdot\|$ is the spectral norm. Then $\overline{F}\leq F(\scrM)$, and the inequality is saturated if each  $\hat{\rho}_j$ is supported in the eigenspace of $\caQ(M_j)$ with the maximum eigenvalue by \rcite{Zhu22}. In view of this fact, $F(\scrM)$ is called the \emph{estimation fidelity} of $\scrM$ henceforth. The definition of the estimation fidelity is still applicable even if $\scrM$ is an incomplete POVM, which means $\sum_j M_j\leq \bbone^{\otimes N}$. 
If there is no restriction on the POVMs that can be performed, then the maximum  estimation fidelity  is $(N+1)/(N+d)$, and optimal POVMs can be constructed from complex projective $t$-designs with $t=N$ \cite{MassP95,BrusM99,HayaHH05,Zhu22}.

\subsection{Properties of the map $\caQ$ and estimation fidelity}

The basic properties of the map $\caQ$ and estimation fidelity $F(\scrM)$ are clarified in \rcite{Zhu22}.  Notably, 
$\|\caQ(M)\|$ and  $F(\scrM)$ are invariant under \emph{symmetric local unitary transformations}, which are associated with unitary operators of the form $U^{\otimes N}$ with $U\in \rmU(\caH)$. Here we introduce some additional results that are relevant to the following discussion. Note that the argument of $\caQ$ is not restricted to POVM elements and is not necessarily Hermitian. In analogy to $P_N$, let $P_N^\caA$ be the projector onto the antisymmetric subspace in $\caH^{\otimes N}$. Let $\caS_N$ be the symmetric group of the $N$ parties associated with the $N$ copies of $\rho$. For each $\sigma\in \caS_N$, let $\bbW_\sigma$ be the unitary operator on $\caH^{\otimes N}$ tied to the permutation $\sigma$. Then 
\begin{align}
	P_N = \frac{1}{N!}\sum_{\sigma\in \caS_N}\bbW_\sigma,\;
	P_N^\caA = \frac{1}{N!}\sum_{\sigma\in \caS_N}\sgn(\sigma)\bbW_\sigma,
\end{align}
where $\sgn(\sigma)=1$ when $\sigma$ is an even permutation and $\sgn(\sigma)=-1$ when $\sigma$ is an odd permutation.  The following lemma is a simple corollary of the definition of $\caQ$ in \eref{eq:Qmap}. 
\begin{lemma}\label{lem:Q}
	Suppose $M\in \caL(\caH^{\otimes N})$; then
	\begin{gather}
		\caQ(\bbW_\sigma M \bbW_\tau) = \caQ(M) \quad \forall \, \sigma,\tau \in \caS_N,\\
		\caQ[(P_{N-1}\otimes \bbone)M(P_{N-1}\otimes \bbone)] = \caQ(M). 
	\end{gather}
	If in addition $N\geq 3$, then 
	\begin{equation}
		\caQ\bigl[\bigl(P_{N-1}^\caA\otimes \bbone\bigr)M\bigl(P_{N-1}^\caA\otimes \bbone\bigr)\bigr] =0.
	\end{equation}  
\end{lemma}
Given any POVM $\scrM=\{M_j\}_j$ on $\caH^{\otimes N}$,  define
\begin{equation}\label{eq:POVMsymAsym}
\begin{aligned}
	\scrM_\caS&:=\{(P_{N-1}\otimes \bbone) M_j (P_{N-1}\otimes \bbone)\}_j,\\ \scrM_\caA&:=\bigl\{\bigl(P_{N-1}^\caA \otimes \bbone\bigr) M_j \bigl(P_{N-1}^\caA \otimes \bbone\bigr)\bigr\}_j.
\end{aligned}
\end{equation} 
Then $\scrM_\caS$ is a POVM on $\Sym_{N-1}(\caH)\otimes \caH$, and $\scrM_\caA$ is a POVM on $\supp(P_{N-1}^\caA)\otimes \caH$ (zero POVM elements can be deleted by default).
The following lemma is a simple corollary of \lref{lem:Q}. 
\begin{lemma}\label{lem:POVMsym}
	Suppose $\scrM$ is a POVM on $\caH^{\otimes N}$. Then 
	\begin{align}
		F(\scrM_\caS)=F(\scrM),\quad F(\scrM_\caA)=0. \label{eq:POVMsym}
	\end{align}
\end{lemma}

\section{Optimal estimation of three parallel spins with genuine collective measurements}
\label{sec:CollectiveOpt}

From now on we turn to the estimation of three parallel spins, which means $d=2$, $N=3$, and  $\caH$ is a single-qubit Hilbert space. In the following discussion, we label the three copies of $\caH$ by $\rmA$, $\rmB$, and  $\rmC$, respectively. 

To benchmark the performance of restricted collective measurements, we first reexamine an optimal estimation strategy when there is no restriction on the measurements that can be performed. In this case,  the maximum  estimation fidelity  is $4/5$ \cite{MassP95,LatoPT98}. Consider the three Pauli operators
\begin{align}
	X=\begin{pmatrix}
		0 &1\\ 
		1 &0
	\end{pmatrix}, \;\; Y=\begin{pmatrix}
		0 &-\rmi\\ 
		\rmi &0
	\end{pmatrix}, \;\; Z=\begin{pmatrix}
		1 &0\\ 
		0 &-1
	\end{pmatrix},
\end{align}
and let $|\psi_j\>$ for $j=1,2,...,6$ 
be the six eigenstates, which form a regular octahedron when represented on the Bloch sphere. Then  an optimal POVM $\scrE$ can be constructed from the following seven POVM elements \cite{LatoPT98}:
\begin{equation}\label{eq:POVMcol}
\begin{aligned}
	E_j &:= \frac{2}{3}|\psi_j\>\<\psi_j|^{\otimes 3},\quad j = 1,2,\ldots,6,\\
	E_7 &:= \bbone - \sum_{j=1}^6 E_j = \bbone - P_3,
\end{aligned}
\end{equation} 
where $P_3$ is the projector onto $\Sym_3(\caH)$. Although this POVM was constructed more than 20 years ago, its intriguing properties have not been fully appreciated.

Let $\Pi=P_2^\caA\otimes \bbone$ and let $\bbW$ be the unitary operator on $\caH^{\otimes 3}$ that is associated with a cyclic permutation. Then $\bbW^2=\bbW^\dag$ and the POVM element $E_7$ can be expressed as follows:
\begin{align}
E_7=\frac{2}{3}\bigl(\Pi +\bbW \Pi\bbW^\dag + \bbW^\dag \Pi\bbW \bigr),
\end{align}
which means $E_7$ is biseparable. So all POVM elements in the optimal POVM $\scrE$ are biseparable. Surprisingly, however, $\scrE$ is not biseparable as shown in the companion paper \cite{ZhouYYH23}. Alternatively, this conclusion also follows from \thref{thm:bisep} below. 

Collective measurements are in general not easy to realize. If we can only perform local measurements on individual copies, then the maximum estimation fidelity is $(3 +\sqrt{3}\lsp)/6$, and the maximum can be attained when Alice, Bob, and Charlie perform Pauli $X$, $Y$, and $Z$ measurements, respectively~\cite{BagaMM05}. Note that the measurement bases of the three parties are mutually unbiased. Recall that two bases $\{|\psi_j\>\}_{j=0}^{d-1}$ and $\{|\varphi_k\>\}_{k=0}^{d-1}$ in $\caH$ are mutually unbiased if $|\<\psi_j|\varphi_k\>|^2=1/d$ for all $j,k$ \cite{Ivan81,WootF89,DurtEBZ10}. Such bases will also be useful for constructing optimal biseparable measurements (in a subtle way), including optimal $2+1$ adaptive measurements, as we shall see later.

\section{Optimal estimation of three parallel spins with restricted collective measurements}
\label{sec:RestrictedOpt}

Although  the maximum estimation fidelity of general collective measurements was clarified a long time ago, the performance of restricted collective measurements is still poorly understood. To fill this gap, here we shall determine  the maximum estimation fidelities of biseparable measurements, $2+1$ adaptive measurements, and $1+2$ adaptive measurements, respectively, in  the estimation of three parallel spins, and construct optimal estimation strategies explicitly.  It turns out that  $2+1$ adaptive measurements can achieve the maximum estimation fidelity of biseparable measurements. As in \sref{sec:CollectiveOpt}, here $\caH$ is a single-qubit Hilbert space.

By symmetry, the maximum estimation fidelity of
$(\rmA\rmB|\rmC)$ separable POVMs on $\caH^{\otimes 3}$ is  identical to the counterpart of $(\rmA\rmC|\rmB)$ separable POVMs and the counterpart of $(\rmB\rmC|\rmA)$ separable POVMs. Moreover, this maximum estimation fidelity is also the maximum estimation fidelity of general biseparable POVMs, given that coarse graining cannot increase the estimation fidelity \cite{Zhu22}. In addition, it suffices to consider rank-1 POVMs to determine the maximum estimation fidelity. 

Furthermore, the maximum estimation fidelity of $(\rmA\rmB|\rmC)$ separable POVMs on $\caH^{\otimes 3}$ is identical to  the maximum estimation fidelity of separable POVMs on $\Sym_2(\caH)\otimes \caH$ thanks to \lref{lem:POVMsym}. If $\scrM$ is an optimal
$(\rmA\rmB|\rmC)$ separable POVM on $\caH^{\otimes 3}$, then $\scrM_\caS$ defined in \eref{eq:POVMsymAsym} is an optimal separable POVM on  $\Sym_2(\caH)\otimes \caH$. Conversely,  if $\scrM$ is an optimal separable POVM   on  $\Sym_2(\caH)\otimes \caH$, then  an
 optimal biseparable POVM on $\caH^{\otimes 3}$ can be constructed  as follows:
\begin{align} 
\scrM\cup \bigl\{P_2^\caA\otimes \bbone\bigr\}. 
\end{align}
Similar remarks apply to $2+1$ adaptive POVMs and $1+2$ adaptive POVMs.

\subsection{Biseparable measurements}

According to the previous discussion, to determine the maximum estimation fidelity of biseparable POVMs on $\caH^{\otimes 3}$, it suffices to consider separable rank-1 POVMs on $\Sym_2(\caH)\otimes \caH$. Suppose $\scrM=\{M_j\}_j$ is such a POVM, then each POVM element $M_j$ has the form
\begin{equation}\label{eq:povm}
	M_j=w_j|\Psi_j\>\<\Psi_j|,\quad |\Psi_j\>= |\Phi_j\>\otimes|\varphi_j\>,
\end{equation}
where $|\Phi_j\>\in \Sym_2(\caH)$ and  $|\varphi_j\>\in \caH$. In addition,
\begin{equation}\label{eq:POVMbs}
	\begin{gathered}
		w_j > 0, \quad \sum_j w_j =6,\\
		\sum_j w_j|\Phi_j\>\<\Phi_j|\otimes|\varphi_j\>\<\varphi_j| = P_2\otimes \bbone,
	\end{gathered}
\end{equation}
where $P_2$ is the projector onto $\Sym_2(\caH)$. 
In conjunction with \eref{eq:EstimationFidDef}, the maximum estimation fidelity of biseparable POVMs can be expressed as follows:
\begin{equation}\label{eq:BisepOpt}
\max_{\scrM}  \frac{1}{120}\sum_j w_j \|\caQ(|\Phi_j\>\<\Phi_j|\otimes|\varphi_j\>\<\varphi_j|)\|,
\end{equation}
where the maximization is subjected to the constraints in \eref{eq:POVMbs}.

By virtue of  \eqsref{eq:POVMbs}{eq:BisepOpt} 
we can derive an analytical formula for the maximum estimation fidelity and clarify the structure of optimal separable  POVMs on  $\Sym_2(\caH)\otimes \caH$
as shown in \lref{lem:bisep} below. On this basis we can further determine the maximum estimation fidelity of biseparable POVMs and clarify the structure of optimal POVMs as shown in \thref{thm:bisep} below. See \aref{app:proofofthm2} for proofs of \lref{lem:bisep} and \thref{thm:bisep}. 

Given a positive integer $t$, recall that a set of unitaries $\caU = \{U_j\}_j$ on $\caH$ is a (unitary) $t$-design \cite{DankCEL09,GrosAE07} if the following equation holds
\begin{equation}
	\frac{1}{|\caU|}\sum_{U_j\in\caU}U_j^{\otimes t}O U_j^{\dag\otimes t} =  \int_{\mathrm{Haar}} dU U^{\otimes t}O U^{\dag\otimes t}
\end{equation}
for all $O\in \caL\bigl(\caH^{\otimes t}\bigr)$. For example, the Clifford group forms a 3-design \cite{Zhu17MC,Webb16}. 
Define
\begin{equation}\label{eq:tPsi}
	\begin{aligned}
		|\tPhi\> &:= \frac{\sqrt{8 + 3\sqrt{7}}}{4}|00\> + \frac{\sqrt{8 - 3\sqrt{7}}}{4}|11\>,\\
		|\tPsi\> &:= |\tPhi\>\otimes|+\>,
	\end{aligned}
\end{equation}
where $|\pm\> = (|0\> \pm |1\>)/\sqrt{2}$ are the eigenstates of $X$ with eigenvalues $\pm1$. Note that  the state $|\tPhi\>$ has concurrence $1/8$ \cite{HillW97}, and $|\pm\>$ are mutually unbiased with respect to the Schmidt basis of $|\tPhi\>$ for each party. In addition, $|\tPhi\>\otimes|+\>$ and $|\tPhi\>\otimes|-\>$ are equivalent under the symmetric local unitary transformation $Z^{\otimes 3}$. 

\begin{proposition}\label{pro:bisepSymPOVM}
	Suppose  $\{U_j\}_{j=1}^m$ is a unitary 2-design on $\caH$, then  $\bigl\{3U_j^{\otimes 2}|\tPhi\>\<\tPhi|U_j^{\dag\otimes 2}/m\bigr\}_{j=1}^m$ is a POVM on $\Sym_2(\caH)$. If in addition $\{U_j\}_{j=1}^m$ is a  3-design, then 
	$\bigl\{6U_j^{\otimes 3}|\tPsi\>\<\tPsi|U_j^{\dag\otimes 3}/m\bigr\}_{j=1}^m$ is a POVM on $\Sym_2(\caH)\otimes \caH$. 
\end{proposition} 

\Pref{pro:bisepSymPOVM}  follows from a similar reasoning that is used to prove \eref{eq:IPsi} in \aref{app:proofofthm2}. It offers a simple way for constructing separable POVMs on $\Sym_2(\caH)\otimes \caH$. Surprisingly, all such POVMs are optimal for estimating three parallel spins among separable POVMs on $\Sym_2(\caH)\otimes \caH$. 

\begin{lemma}\label{lem:bisep}
	Suppose $\scrM=\{M_j\}_j$ is a separable POVM on $\Sym_2(\caH)\otimes \caH$. Then 
	\begin{equation}\label{eq:bisepSymUB}
		F(\scrM)\leq F_\bs:=\frac{1}{2} + \frac{\sqrt{22}}{16},
	\end{equation}
	and the upper bound is saturated iff each $M_j/\tr(M_j)$  is a pure  state that is equivalent to $|\tPsi\>\<\tPsi|$ under a symmetric local unitary transformation.
\end{lemma}
Note that an optimal separable POVM $\scrM$ on $\Sym_2(\caH)\otimes \caH$ is automatically rank-1. Moreover, all normalized POVM elements of $\scrM$ are equivalent to each other under symmetric local unitary transformations and thus have the same entanglement structure. Notably, $\tr_\rmC(M_j)/\tr(M_j)$ always has concurrence $1/8$; in addition, the eigenbasis of $\tr_{\rmA\rmB}(M_j)$ is  mutually unbiased with the Schmidt basis of $\tr_\rmC(M_j)$ for each party. Now the appearance of mutually unbiased bases is more subtle compared with the optimal strategies based on local projective measurements \cite{BagaMM05}. 
 Thanks to  \lref{lem:POVMsym}, \eref{eq:bisepSymUB} still holds if instead $\scrM$ is a POVM on $\caH^{\otimes 3}$ that is $(\rmA\rmB|\rmC)$ separable. 
Also,  by symmetry the same conclusion holds if $\scrM$ is  $(\rmA\rmC|\rmB)$ separable or $(\rmB\rmC|\rmA)$ separable.

\begin{theorem}\label{thm:bisep}
	Suppose $\scrM=\{M_j\}_j$ is a biseparable POVM on $\caH^{\otimes 3}$; then  $F(\scrM)\leq F_\bs$.	If in addition $\scrM$ is $(\rmA \rmB|\rmC)$ separable, then the maximum estimation fidelity $F_\bs$ can be attained iff $(P_2\otimes \bbone)M_j(P_2\otimes \bbone)$ for each $j$ is proportional to a quantum state that is equivalent to $|\tPsi\>\<\tPsi|$ under a symmetric local unitary transformation.
\end{theorem}
Note that biseparable measurements can achieve a higher estimation fidelity than local measurements, but there is a fundamental gap from genuine collective measurements, as summarized in \tref{tab:summary}.

\begin{table}[t]
	\centering
	\caption{\label{tab:summary}Maximum estimation fidelities of four different types of measurements.}    
	\renewcommand\arraystretch{2}
	\begin{tabular}{ l  l }
		\hline
		\hline
		Measurement & Maximum estimation fidelity \\
		\hline 
		Local & $\frac{1}{2} + \frac{\sqrt{3}}{6} \approx 0.78868$ \cite{BagaMM05}\\ 

		$1+2$ adaptive &  $\frac{1}{2} + \frac{11 + \sqrt{41}}{60} \approx 0.79005$  \\

		\makecell[l]{Biseparable\\
			(or $2+1$ adaptive)} & $\frac{1}{2} + \frac{\sqrt{22}}{16} \approx 0.79315$ \\
		Genuine collective & $\frac{4}{5}$ \cite{MassP95}\\
		\hline
		\hline
	\end{tabular}
\end{table}

\subsection{$2+1$ adaptive measurements}

Here we show that $2+1$ adaptive measurements can achieve the maximum estimation fidelity of biseparable measurements. 

\begin{theorem}\label{thm:2+1}
Suppose $\scrM = \{M_j\}_j$ is a $2+1$ adaptive POVM on $\caH^{\otimes 3}$   with respect to the bipartition $(\rmA \rmB|\rmC)$. Then 
	\begin{equation}\label{eq:2+1}
		F(\scrM)\leq F_{2\rightarrow 1}:=	\frac{1}{2} + \frac{\sqrt{22}}{16};
	\end{equation}
	the upper bound is saturated  iff $(P_2\otimes \bbone)M_j(P_2\otimes \bbone)$ for each $j$ is proportional to a quantum state that is equivalent to $|\tPsi\>\<\tPsi|$ under a symmetric local unitary transformation.
\end{theorem}

\Thref{thm:2+1} is a simple corollary of  \thref{thm:bisep} given that $2+1$ adaptive POVMs are biseparable. A direct proof of \eref{eq:2+1} can be found in \aref{app:2+1}.

By virtue of \pref{pro:bisepSymPOVM} and \thref{thm:2+1}, it is easy to construct optimal $2+1$ adaptive POVMs that can attain the maximum estimation fidelity $F_{2\rightarrow 1}$. Suppose 	$\{U_j\}_{j=1}^m$ is a unitary 2-design on $\caH$, then $\bigl\{3U_j^{\otimes 2}|\tPhi\>\<\tPhi|U_j^{\dag\otimes 2}/m\bigr\}_{j=1}^m$ is a POVM on $\Sym_2(\caH)$ by \pref{pro:bisepSymPOVM}, where $|\tPhi\>$ is defined in \eref{eq:tPsi}; accordingly,
\begin{equation}\label{eq:POVMsym2}
	\biggl\{\frac{3}{m}U_j^{\otimes 2}|\tPhi\>\<\tPhi|U_j^{\dag\otimes 2}\biggr\}_{j=1}^m \cup \bigl\{P_2^\caA\bigr\}
\end{equation}
is a POVM on $\caH^{\otimes 2}$. Now, an optimal $2+1$ adaptive POVM can be realized as follows: Alice and Bob first perform the POVM in \eref{eq:POVMsym2} and send the outcome to Charlie; if they obtain outcome $3U_j^{\otimes 2}|\tPhi\>\<\tPhi|U_j^{\dag\otimes 2}/m$ (note that the outcome $P_2^\caA$ can never occur), then Charlie performs the projective measurement on the eigenbasis of $U_jXU_j^\dag$, which is mutually unbiased with the Schmidt basis of $U_j^{\otimes 2}|\tPhi\>\<\tPhi|U_j^{\dag\otimes 2}$ for each party.

An explicit optimal POVM can be constructed by virtue of the single-qubit Clifford group $\Cl_1$ or a suitable subgroup. Recall that the Pauli group is generated by the three Pauli operators $X,Y,Z$. 
The Clifford group $\Cl_1$ is the normalizer of the Pauli group and is a unitary 3-design \cite{Zhu17MC,Webb16}. Up to overall phase factors, it is generated by the Hadamard gate $H$ and phase gate~$S$, where
\begin{align}\label{eq:HS}
	H=\frac{1}{\sqrt{2}}\begin{pmatrix}
		1 &1\\ 
		1 &-1
	\end{pmatrix} , \quad S=\begin{pmatrix}
		1 &0\\ 
		0 &\rmi
	\end{pmatrix}. 
\end{align}
Let $V=HS$ and let $\caG$ be the group generated by $X$ and $V$, then $\caG$ is a subgroup of $\Cl_1$ that contains the Pauli group. Let 
\begin{align}
	\overline{\caG}:=\bigl\{\bbone, V, V^2\bigr\}\times \{\bbone, X, Y, Z\}; 
\end{align}
then $\overline{\caG}$ can be identified as the quotient group of $\caG$ after identifying operators that are proportional to each other. Thanks to a criterion derived in \rcite{Zhu17MC}, it is straightforward to verify that  $\overline{\caG}$ forms a unitary 2-design.

By virtue of $\overline{\caG}$ we can  construct an optimal  $2+1$ adaptive POVM as explained above. Now the construction can be simplified slightly because the state $|\tPhi\>$ is
stabilized by the operator $Z^{\otimes 2}$. 
To be specific, the group $\overline{\caG}$ can be replaced by the following subgroup: 
\begin{equation}
	\overline{\caG}_2:=\bigl\{\bbone, V, V^2\bigr\}\times\{\bbone,X\}.
\end{equation}
Let
\begin{equation}
	\!\!\scrM_2:= \Bigl\{\frac{1}{2}U^{\otimes 2}|\tPhi\>\<\tPhi|U^{\dag\otimes 2}\Big| U\in \overline{\caG}_2\Bigr\}\cup\bigl\{P_2^\caA\bigr\};
\end{equation}
then $\scrM_2$ is a POVM on $\caH^{\otimes 2}$ although $\overline{\caG}_2$ is not a 2-design. 
On this basis, an optimal $2+1$ adaptive POVM can be realized as follows: Alice and Bob first perform the POVM $\scrM_2$ and send the outcome to Charlie. If they obtain outcome $U^{\otimes 2}|\tPhi\>\<\tPhi|U^{\dag\otimes 2}/2$, then Charlie performs the projective measurement on the eigenbasis of $UXU^\dag$. 
The resulting POVM has 13 POVM elements and can  be expressed as follows:
\begin{align}\label{eq:OptPOVMbs2+1}
	\scrM_{2\rightarrow 1}:=\Bigl\{\frac{1}{2}U^{\otimes 3}|\tPsi\>\<\tPsi|U^{\dag\otimes 3}\,\Big|\, U\in \overline{\caG} \Bigr\}\cup \bigl\{P_2^\caA\otimes \bbone\bigr\}, 
\end{align}
where $|\tPsi\>$ is defined in \eref{eq:tPsi}. By virtue of \thref{thm:bisep} it is also straightforward to verify that $\scrM_{2\rightarrow 1}$ is an optimal biseparable 
POVM on $\caH^{\otimes 3}$ (although $\overline{\caG}$ is not a 3-design).

\subsection{$1+2$ adaptive measurements}

In this section, we determine the maximum estimation fidelity of $1+2$ adaptive measurements and devise an optimal strategy. 

\begin{theorem}	\label{thm:1+2}
The maximum estimation fidelity of  $1+2$ adaptive POVMs on $\caH^{\otimes 3}$ is
\begin{equation}\label{eq:1+2}
F_{1\rightarrow 2}:=	\frac{1}{2} + \frac{11 + \sqrt{41}}{60}.
	\end{equation}
\end{theorem}
\Thref{thm:1+2} is proved in \aref{app:1+2}.
Here we construct an optimal POVM that can attain the maximum estimation fidelity in \eref{eq:1+2}. Let 
\begin{equation}\label{eq:pSPhiW}
\begin{aligned}
	p & := \frac{47 - 3\sqrt{41}}{216},\quad  |\bbS\> :=\frac{|01\> + |10\>}{\sqrt{2}},
	\\
	|\Upsilon\> &:= \sqrt{\frac{1 - 3p}{3-3p}}\,|00\> + \sqrt{\frac{1}{3-3p}}\left(|\bbS\> + |11\>\right),\\
	 W_j &:=(S\otimes S)^j,\quad j=0,1,2,3,
\end{aligned}
\end{equation}
where $S$ is the phase gate defined in \eref{eq:HS}. Then we can construct two POVMs $\scrK_0, \scrK_1$ on $\caH^{\otimes 2}$ as follows:
\begin{equation}\label{eq:scrE01}
\begin{aligned}
	K_j&:=\frac{3-3p}{4}W_j|\Upsilon\>\<\Upsilon|W_j^\dag,\quad j=0,1,2,3,\\
	K_4 &:=3p |00\>\< 00|,\quad K_5:=P_2^\caA,\\
	\scrK_0 & := \{K_j\}_{j=0}^5,\quad \scrK_1:=\bigl\{X^{\otimes 2}K_jX^{\otimes 2}\bigr\}_{j=0}^5.
\end{aligned}
\end{equation}
On this basis, we can construct an optimal $1+2$ adaptive POVM:
\begin{align}\label{eq:OptPOVM1+2}
\scrM_{1\rightarrow 2} :=\{\scrK_0\otimes |0\>\<0|\}\cup \{\scrK_1\otimes |1\>\<1|\}. 
\end{align}
This POVM can be realized as follows: Charlie first performs the $Z$-basis measurement on his qubit and sends the measurement outcome to Alice and Bob.  If the outcome is 0, then Alice and Bob perform the POVM $\scrK_0$ on their qubits; if the outcome is~1, then they perform the POVM $\scrK_1$ instead. Note that 
not all normalized POVM elements of $\scrM_{1\rightarrow 2}$ supported in the subspace $\Sym_2(\caH)\otimes \caH$
are equivalent under symmetric local unitary transformations, in sharp contrast with optimal $2+1$
adaptive POVMs as clarified in \thref{thm:2+1}. This distinction further highlights the importance of communication direction.

\bigskip

\section{Conclusion}
\label{sec: conclusion}

In this work we introduced rigorous definitions of biseparable measurements and genuine collective measurements, thereby 
setting the stage for exploring the rich territory of multipartite nonclassicality in quantum measurements instead of quantum states. By virtue of a simple estimation problem, we established a rigorous fidelity gap between biseparable measurements and genuine collective measurements. Moreover, we showed that the maximum estimation fidelity of biseparable measurements can be attained 
by  $2+1$ adaptive measurements, but not by  $1+2$ adaptive measurements. Optimal estimation protocols in all these settings are constructed explicitly. Our work shows that quantum measurements in the multipartite setting may exhibit a rich hierarchy of nonclassical phenomena, which offer exciting opportunities for future studies.

\section*{Acknowledgments}
The work at Fudan University is supported by the National Key Research and Development Program of China
(Grant No. 2022YFA1404204), Shanghai Science and Technology Innovation Action Plan (Grant No. 24LZ1400200), Shanghai Municipal Science and Technology Major Project
(Grant No. 2019SHZDZX01), Innovation Program for Quantum Science and Technology (Grant No. 2024ZD0300101), 
and the National Natural Science Foundation of China (Grant No. 92165109). The work at the University of Science and Technology of China is supported by the Innovation
Program for Quantum Science and Technology (Grant No.2023ZD0301400), the National Natural Science Foundation
of China (Grants No. 62222512, No. 12104439, and No.
12134014), the Anhui Provincial Natural Science Foundation
(Grant No. 2208085J03), and USTC Research Funds of the
Double First-Class Initiative (Grants No. YD2030002007 and
No. YD2030002011).
\onecolumngrid

\bigskip
\appendix

\counterwithin{lemma}{section}
\renewcommand{\thelemma}{\Alph{section}\arabic{lemma}} 
\setcounter{lemma}{0} 

\section{Proofs of \lref{lem:bisep} and \thref{thm:bisep}}
\label{app:proofofthm2}

 In this and all following Appendixes, we prove the key results
on optimal estimation of three parallel spins presented in the main text, namely, \lref{lem:bisep}, \thsref{thm:bisep} and \ref{thm:1+2}, and \eref{eq:2+1}. Throughout the Appendixes, we assume that $\caH$ is a single-qubit Hilbert space. 

\subsection{Auxiliary results}
Suppose $a,b,c,x,y,\phi$ are real numbers.  Define
\begin{align}
\eta(a,b,c,\phi)&:= 15 a^2 + 10b^2 + 5c^2  + \sqrt{8b^2|3a + 2c\rme^{\rmi\phi}|^2 + (9a^2 + 2b^2 - c^2)^2},   \label{eq:etaabcphiDef}\\
\eta(a,b,c)&:= 15 a^2 + 10b^2 + 5c^2  + \sqrt{8b^2(3|a| + 2|c|)^2 + (9a^2 + 2b^2 - c^2)^2}, \label{eq:etaabcDef}\\
f(x,y)& := \sqrt{4(1-x)(3\sqrt{x + y} + 2\sqrt{x - y}\lsp)^2 + (2x + 5y + 2)^2}. \label{eq:fxy}
\end{align}
Given any state $|\Psi \>$ in $\caH^{\otimes 3}$, define
\begin{equation}
\caT_\Psi := \int_{\mathrm{Haar}}\rmd U U^{\otimes 3} |\Psi\>\<\Psi| U^{\dag\otimes 3},   
\end{equation}
where the integration is taken over the normalized Haar measure on the unitary group $\rmU(\caH)$. 
Alternatively,  the integration can be replaced by summation over a unitary 3-design \cite{DankCEL09,GrosAE07}.

\begin{lemma}\label{lem:q_abc}
Suppose $|\Phi \> = a|00\> + b|\bbS\> + c\rme^{\rmi\phi}|11\>\in \caH^{\otimes 2}$ and $|\Psi\>=|\Phi\>\otimes |0\>\in \caH^{\otimes 3}$, where $a,b,c,\phi$ are real numbers and $|\bbS\>=(|01\>+|10\>)/\sqrt{2}$. Then 
\begin{gather}
6\tr(P_3 |\Psi\>\<\Psi|) = 6a^2+4b^2+2c^2, \label{eq:P3Psi0}\\ 
\caT_\Psi = \frac{3a^2+2b^2+c^2}{12}P_3+\frac{b^2+2c^2}{6}(P_2\otimes \bbone-P_3),   \label{eq:IPsi} \\
\|\caQ(|\Psi\>\<\Psi|)\| = \eta(a,b,c,\phi)\le \eta(a,b,c)=15 a^2 + 10b^2 + 5c^2  + f\bigl(a^2+c^2, a^2-c^2\bigr),  \label{eq:QPsinorm2}
\end{gather}
and the last inequality is saturated when $\phi=0$.
\end{lemma}

When $a^2 = c^2$, \eref{eq:IPsi} implies that
\begin{equation}
    \caT_\Psi = \frac{1}{6}P_2\otimes \bbone.
\end{equation}
In this case, if $\{U_j\}_{j=1}^m$ is a unitary 3-design, then $\bigl\{6U_j^{\otimes 3}|\Psi\>\<\Psi|U_j^{\dag\otimes 3}/m\bigr\}_{j=1}^m$ is a POVM on $\Sym_2(\caH)\otimes \caH$. 
This conclusion is useful for constructing optimal biseparable POVMs. 

\begin{proof}[Proof of \lref{lem:q_abc}]
\Eref{eq:P3Psi0} can be verified by straightforward calculation. By Schur-Weyl duality, $\caT_\Psi$ is a linear combination of $\bbW_\sigma$ for $\sigma\in \caS_3$. In addition $\caT_\Psi=\bbW_{(12)}\caT_\Psi= \caT_\Psi\bbW_{(12)}$, where $(1 2)$ denotes the transposition of the first two parties.  So $\caT_\Psi$ can only be a linear combination of $P_3$ and $P_2\otimes \bbone$. Now \eref{eq:IPsi} is a simple corollary of \eref{eq:P3Psi0}. 

Next, straightforward calculation yields
\begin{align}
 \caQ(|\Psi\>\<\Psi|) &= \bigl(15a^2 + 10b^2 +5c^2\bigr)\bbone  + 2\sqrt{2}\lsp b[3a + 2c \cos (\phi)] X + 4\sqrt{2}\lsp b c \sin(\phi) Y + \bigl(9a^2 + 2b^2 - c^2\bigr)Z,
    \end{align}
which implies \eref{eq:QPsinorm2} given    the definitions in Eqs.~\eqref{eq:etaabcphiDef}-\eqref{eq:fxy}.     
\end{proof}

\begin{lemma}\label{lem:fxyUB}
	Suppose $0 \le x \le 1$ and $-x \le y \le x$. Then the function $f(x,y)$  defined in \eref{eq:fxy} satisfies
	\begin{equation}\label{eq:fxyUB}
	f(x,y) \le \frac{5}{2}\sqrt{\frac{11}{2}} + 8\sqrt{\frac{2}{11}}\,y;
	\end{equation}
and the inequality is saturated iff $x=9/16$ and $y=0$. 	
\end{lemma}
\begin{proof}
	Due to continuity, it suffices to prove \eref{eq:fxyUB} when $0<x<1$ and $-x<y<x$. 		
	Direct calculation yields	
	\begin{gather}
	\frac{5}{2}\sqrt{\frac{11}{2}} + 8\sqrt{\frac{2}{11}}\,y\geq  \frac{5}{2}\sqrt{\frac{11}{2}} - 8\sqrt{\frac{2}{11}}>0,\quad 
	\left(\frac{5}{2}\sqrt{\frac{11}{2}} + 8\sqrt{\frac{2}{11}}\,y\right)^2 - f^2(x,y) = r(x,z),
	\end{gather}
	where  $z := y^2<x^2$ and
	\begin{equation}
	r(x,z): = \frac{243}{8} + 48x^2 - \frac{147z}{11} + 48(x-1)\sqrt{x^2 - z} - 60x.
	\end{equation}
	The partial derivative of $r(x,z)$ over $z$ reads
	\begin{equation}
	\frac{\partial r}{\partial z} = \frac{24(1-x)}{\sqrt{x^2 - z}} - \frac{147}{11}.
	\end{equation}
	If  $0<x< 88/137$, then this derivative is always positive. Therefore, 
	\begin{equation}\label{eq:fxyUBproof}
	r(x,z) \ge r(x,0) = \frac{3}{8}(9-16x)^2\geq 0,
	\end{equation}
	which implies \eref{eq:fxyUB}. 
	If instead $88/137\leq x< 1$, then the partial derivative $\partial r/\partial z$ has a unique zero,  denoted by $z_0$ henceforth. In addition, $z_0$  satisfies the equation 
	\begin{equation}
	\sqrt{x^2 - z_0} = \frac{264}{147}(1-x),
	\end{equation}
	which means 
	\begin{align}
	z_0=\frac{-5343x^2+15488x-7744}{2401}. 
	\end{align}
	Therefore,
	\begin{equation}
	r(x,z) \ge r(x,z_0)=-\frac{3(12168x^2-37664x+18293)}{4312}\geq \frac{91875}{150152}>0,
	\end{equation}
which implies \eref{eq:fxyUB}. Here the last inequality follows from the assumption $88/137\leq x< 1$.

If the inequality in \eref{eq:fxyUB} is saturated,  then $0<x< 88/137$ and the two inequalities in 
\eref{eq:fxyUBproof} are saturated simultaneously, which means $x=9/16$ and $y=z=0$, in which case the inequality in \eref{eq:fxyUB} is indeed saturated. 
\end{proof}

\begin{lemma}\label{lem:etaSumUB}
	Suppose $w_j,a_j,b_j,c_j,\phi_j$ are nonnegative real numbers that satisfy the following conditions:
	\begin{align}
		w_j>0, \quad \phi_j < 2\pi, \quad a_j^2 + b_j^2 + c_j^2 = 1\quad \forall \, j;\quad  \sum_j w_j = 6, \quad \sum_j w_j a_j^2 = \sum_j w_j c_j^2. 
	\end{align}
	Then 
	\begin{equation}\label{eq:etaSumUB}
		\sum_j w_j \eta(a_j,b_j,c_j,\phi_j) \le 60 + \frac{15}{2}\sqrt{22},
	\end{equation}
	and the inequality is saturated iff $(a_j,b_j,c_j,\phi_j) = (3\sqrt{2}/8,\sqrt{7}/4,3\sqrt{2}/8,0)$ for all $j$.
\end{lemma}

\begin{proof}
	Let 
	\begin{align} 
		p_j =\frac{w_j}{6},\quad x_j = a_j^2+c_j^2,\quad y_j = a_j^2-c_j^2;
	\end{align}
	then  $\sum_j p_j=1$ and $\sum_j p_jy_j=0$ by assumption. By virtue of  \lsref{lem:q_abc} and \ref{lem:fxyUB} we can deduce that
	\begin{equation}\label{eq:bsProof}	
		\begin{aligned}
			\sum_j w_j \eta(a_j,b_j,c_j,\phi_j)&\leq \sum_j w_j \eta(a_j,b_j,c_j) =6\sum_{j} p_j[10 + 5y_j + f(x_j,y_j)]=60+6\sum_{j} p_jf(x_j,y_j), \\
			\sum_{j} p_jf(x_j,y_j) &\le \sum_j p_j \left(\frac{5}{2}\sqrt{\frac{11}{2}} + 8\sqrt{\frac{2}{11}}\, y_j\right) = \frac{5}{4}\sqrt{22}, 	
		\end{aligned}
	\end{equation}
	which imply \eref{eq:etaSumUB}. 
	
	If $(a_j,b_j,c_j,\phi_j) = (3\sqrt{2}/8,\sqrt{7}/4,3\sqrt{2}/8,0)$ for all $j$, then 
	\begin{align}
		\eta(a_j,b_j,c_j,\phi_j)=10+\frac{5\sqrt{22}}{4}\quad \forall \, j
	\end{align}
	according to the definition in \eref{eq:etaabcphiDef}, so the inequality in \eref{eq:etaSumUB} is saturated given that $\sum_j w_j=6$ by assumption. Conversely, if the equality in
	\eref{eq:etaSumUB} is saturated, then the two inequalities in \eref{eq:bsProof} are saturated. According to \lref{lem:fxyUB}, the saturation of the second inequality  in \eref{eq:bsProof} implies that $(x_j,y_j)=(9/16,0)$ for all $j$, that is, $(a_j, b_j,c_j)=(3\sqrt{2}/8,\sqrt{7}/4,3\sqrt{2}/8)$ for all $j$. Now, according to the definitions in \eqsref{eq:etaabcphiDef}{eq:etaabcDef}, the saturation of the first inequality  in \eref{eq:bsProof} implies that $\phi_j=0$ for all $j$. Therefore, the inequality in \eref{eq:etaSumUB} is saturated iff $(a_j,b_j,c_j,\phi_j) = (3\sqrt{2}/8,\sqrt{7}/4,3\sqrt{2}/8,0)$ for all $j$, which completes the proof of \lref{lem:etaSumUB}.
\end{proof}

\subsection{Proof of \lref{lem:bisep}}
To start with, we first assume that $\scrM$ is a rank-1 POVM; then $M_j$ can be expressed as $M_j=w_j|\Psi_j\>\<\Psi_j|$ with $|\Psi_j\>\in \Sym_2(\caH)\otimes \caH$, $0<w_j\leq 1$, and $\sum_j w_j=6$.  In addition, 
each $|\Psi_j\>$ can be expressed as follows: 
\begin{align}
	|\Psi_j\>=U_j^{\otimes 3} \bigl[\bigl(a_j|00\> + b_j|\bbS\> + c_j\rme^{\rmi \phi_j}|11\>\bigr)\otimes |0\>\bigr],
\end{align}
where $U_j \in \rmU(\caH)$, 
$a_j, b_j,c_j\geq0$, $a_j^2+b_j^2+c_j^2=1$, and $0\leq\phi_j<2\pi$. By virtue of \lref{lem:q_abc} we can deduce that 
\begin{align}
	F(\scrM)&=\frac{1}{120}\sum_j\|\caQ(w_j|\Psi_j\>\<\Psi_j|)\|=
	\frac{1}{120}\sum_j w_j \eta(a_j,b_j,c_j,\phi_j),\\
	4=\sum_j w_j\tr\left(P_3 |\Psi_j\>\<\Psi_j|\right)&=\sum_j \frac{w_j\bigl(3a_j^2+2b_j^2+c_j^2\bigr)}{3}=\sum_j\frac{w_j\bigl(2+a_j^2-c_j^2\bigr)}{3}=4+\sum_j \frac{w_j\bigl(a_j^2-c_j^2\bigr)}{3},
\end{align}
which means $\sum_j w_j a_j^2=\sum_j w_j c_j^2$. In conjunction with \lref{lem:etaSumUB} we can deduce that
\begin{equation}\label{eq:constbound}
	F(\scrM) = \frac{1}{120}\sum_j w_j\eta(a_j,b_j,c_j,\phi_j) \le \frac{1}{2} + \frac{\sqrt{22}}{16}=F_\bs,
\end{equation}
which confirms the inequality in \eref{eq:bisepSymUB}.

Next, we clarify the conditions under which the inequality in \eref{eq:bisepSymUB} is saturated. If each normalized POVM element $M_j/\tr(M_j)=|\Psi_j\>\<\Psi_j|$ is equivalent to $|\tPsi\>\<\tPsi|$ under a symmetric local unitary transformation, then, according to  \eref{eq:EstimationFidDef}, we have
\begin{align}
	F(\scrM)=\frac{1}{120}\sum_j \|\caQ(M_j)\|=\frac{1}{120}\sum_j w_j \|\caQ(|\Psi_j\>\<\Psi_j|)\|=\frac{1}{20} \|\caQ(|\tPsi\>\<\tPsi|)\|=\frac{1}{2} + \frac{\sqrt{22}}{16},
\end{align}
where the last equality can be verified by straightforward calculation. In this case,
the inequality in \eref{eq:bisepSymUB} is indeed saturated. Conversely, if the inequality in \eref{eq:bisepSymUB} is  saturated, that is, $F(\scrM) = 1/2 + \sqrt{22}/16$, then the inequality in \eref{eq:constbound} is saturated. 
By virtue of \lref{lem:etaSumUB}  we can deduce that  $(a_j, b_j,c_j,\phi_j)=(3\sqrt{2}/8,\sqrt{7}/4,3\sqrt{2}/8,0)$ for all $j$, which implies that
\begin{equation}
	|\Psi_j\rangle = U_j^{\otimes 3} \left[\left(\frac{3\sqrt{2}}{8}|00\rangle + \frac{\sqrt{7}}{4}|\bbS\rangle + \frac{3\sqrt{2}}{8}|11\rangle\right)\otimes|0\rangle\right] = (U_j H)^{\otimes 3}|\tPsi\> \quad \forall \, j. 
\end{equation}
Therefore, each $M_j/\tr(M_j)=|\Psi_j\>\<\Psi_j|$ is equivalent to $|\tPsi\>\<\tPsi|$ under a symmetric local unitary transformation. 

Next, suppose $\scrM$ is not a rank-1 POVM. Then $\scrM$ is a coarse graining of a rank-1 POVM, so the inequality in \eref{eq:bisepSymUB} still holds given that coarse graining cannot increase the estimation fidelity \cite{Zhu22}. In addition, $\scrM$ has at least one POVM element, say $M_1$, that has  rank at least 2. Consequently, the support of $M_1$ contains at least one pure state that is not equivalent to  $|\tPsi\>\<\tPsi|$ under  symmetric local unitary transformations. Therefore, $\scrM$ can be expressed as a coarse graining of a rank-1 POVM $\tscrM$ at least one POVM element of which is not  equivalent to $|\tPsi\>\<\tPsi|$ under  symmetric local unitary transformations, which means
$F(\scrM)\leq F(\tscrM)<F_\bs$. In other words, the inequality in \eref{eq:bisepSymUB} cannot be saturated whenever $\scrM$ is not a rank-1 POVM. This observation completes the proof of \lref{lem:bisep}.

\subsection{Proof of \thref{thm:bisep}}

By assumption $\scrM$  is a coarse graining of a POVM of the form 
\begin{align}
	\scrK=p_3\scrK_{(\rmA\rmB|\rmC)} \sqcup p_2\scrK_{(\rmA\rmC|\rmB)}\sqcup p_1\scrK_{(\rmB\rmC|\rmA)},
\end{align}
where $p_1, p_2, p_3\geq 0$, $p_1+p_2+p_3=1$, and the three POVMs $\scrK_{(\rmA\rmB|\rmC)}$, $\scrK_{(\rmA\rmC|\rmB)}$, and  $\scrK_{(\rmB\rmC|\rmA)}$ are  $(\rmA\rmB|\rmC)$ separable, $(\rmA\rmC|\rmB)$ separable, and $(\rmB\rmC|\rmA)$ separable, respectively. Therefore, 
\begin{align}
	F(\scrM)\leq F(\scrK)= p_3 F(\scrK_{(\rmA\rmB|\rmC)}) + p_2F(\scrK_{(\rmA\rmC|\rmB)})+ p_1F(\scrK_{(\rmB\rmC|\rmA)})\leq F_\bs=\frac{1}{2} + \frac{\sqrt{22}}{16}. 
\end{align}
Here the first inequality holds because coarse graining cannot increase the estimation fidelity \cite{Zhu22}, and  the second inequality follows from 
\lsref{lem:POVMsym} and \ref{lem:bisep}. 

If in addition $\scrM = \{M_j\}_j$ is $(\rmA \rmB|\rmC)$ separable, then  $(P_2\otimes \bbone)\scrM(P_2\otimes \bbone)$ is a separable POVM on $\Sym_2(\caH)\otimes \caH$. According to \lref{lem:bisep},  the maximum estimation fidelity $F_\bs$ can be attained iff $(P_2\otimes \bbone)M_j(P_2\otimes \bbone)$ for each $j$ is proportional to a quantum state that is equivalent to $|\tPsi\>\<\tPsi|$ under a symmetric local unitary transformation, which completes the proof of \thref{thm:bisep}. 

\section{Direct proof of \eref{eq:2+1}}
\label{app:2+1}

\setcounter{lemma}{0} 

\subsection{Auxiliary results}

\begin{lemma}\label{lem:Sym2CanonicalForm}
	Suppose $|\Phi\>\in \Sym_2(\caH)$. Then there exists $U\in \rmU(\caH)$ and $\xi \in [0,\pi/2]$ such that
	\begin{equation}\label{eq:Sym2CanonicalForm}
		U^{\otimes 2}|\Phi\> =\cos\frac{\xi}{2}|00\> + \sin\frac{\xi}{2}|11\>.
	\end{equation}
\end{lemma}

\begin{proof}
	By assumption $|\Phi\>$ can be expressed as follows:
	\begin{equation}
		|\Phi\> = W^{\dag\otimes 2}(a|00\> + b|\bbS\> + c \rme^{\rmi\chi}|11\>),\quad a,b,c \ge 0,\ a^2 + b^2 + c^2=1, \ \chi \in [0,2\pi),
	\end{equation}
	where $W\in \rmU(\caH)$.
	Consider a unitary operator of the form 
	\begin{equation}
		U_1(\theta,\phi) = \begin{pmatrix}
			\cos\theta & \sin\theta \rme^{\rmi\phi} \\
			-\sin\theta \rme^{-\rmi\phi} & \cos\theta
		\end{pmatrix}.
	\end{equation}
    Apply $[U_1(\theta,\phi)W]^{\otimes 2}$ on $|\Phi\>$ yields
	\begin{equation}
		[U_1(\theta,\phi)W]^{\otimes 2}|\Phi\> = u(\theta,\phi)|00\> + v(\theta,\phi)|\bbS\> + w(\theta,\phi)|11\>,\\
	\end{equation}
	where
	\begin{equation}
		\begin{aligned}
			u(\theta,\phi) &= a\cos^2\theta + \frac{1}{\sqrt{2}}b\rme^{\rmi\phi} \sin 2\theta + c\rme^{2\rmi(\phi + \chi)}\sin^2\theta,\\
			v(\theta,\phi) &= b\cos(2\theta) + \frac{1}{\sqrt{2}}\bigl[c\rme^{\rmi(\phi + \chi)} - a\rme^{-\rmi\phi}\bigr]\sin(2\theta)\\
			&=b \cos(2\theta) + \frac{1}{\sqrt{2}}\left[c \cos(\phi + \chi) - a \cos(\phi)\right] \sin(2\theta) + \frac{\rmi}{\sqrt{2}}\left[c \sin(\phi + \chi) + a \sin(\phi)\right] \sin(2\theta),\\
			w(\theta,\phi) &= c\rme^{\rmi\chi}\cos^2\theta - \frac{1}{\sqrt{2}}b\rme^{-\rmi\phi}\sin(2\theta) + a\rme^{-\rmi2\phi}\sin^2\theta.
		\end{aligned}
	\end{equation}
	Let $\phi_0$ be a solution of the equation
	\begin{equation}
		c \sin(\phi_0 + \chi) + a \sin(\phi_0) = 0,
	\end{equation}
	and $\theta_0$ a solution of the equation
	\begin{equation}
		b \cos(2\theta_0) + \frac{1}{\sqrt{2}}\left[c \cos(\phi_0 + \chi) - a \cos\phi_0\right] \sin(2\theta_0) = 0.
	\end{equation}
	Then $v(\theta_0,\phi_0)=0$ and
	\begin{equation}
		[U_1(\theta_0,\phi_0)W]^{\otimes 2}|\Phi\> = u(\theta_0,\phi_0) |00\> + w(\theta_0,\phi_0) |11\>.
	\end{equation}
	Now it is easy to find a diagonal unitary operator $U_2$ (with respect to the computational basis) such that $U_2^{\otimes 2}[u(\theta_0,\phi_0) |00\> + w(\theta_0,\phi_0) |11\>]=\cos(\xi/2)|00\> + \sin(\xi/2)|11\>$ with $\xi \in [0,\pi/2]$. Let $U=U_2 U_1(\theta_0,\phi_0) W$, then \eref{eq:Sym2CanonicalForm} holds, which completes the proof of \lref{lem:Sym2CanonicalForm}.
\end{proof}

\begin{lemma}\label{lem:proofofthm2+1}
	Suppose $|\Phi\>\in \Sym_2(\caH)$ and $\scrM = \{M_j\}_j$ is a POVM on $\caH$. Then 
	\begin{equation}\label{eq:PsiPOVM2+1}
		\sum_j \|\caQ(|\Phi\>\<\Phi|\otimes M_j)\| \le 20 + \frac{5\sqrt{22}}{2},
	\end{equation}
	and the upper bound is saturated when    
	\begin{equation}\label{eq:PsiPOVM2+1Saturate}
		|\Phi\>= \cos\frac{\xi_0}{2}|00\> + \sin\frac{\xi_0}{2}|11\>,\quad \xi_0 := \arcsin\left(1/8\right),\quad  \scrM = \{|+\>\< +|, |-\>\< -|\}. 
	\end{equation}
\end{lemma}

\begin{proof}
	Thanks to \lref{lem:Sym2CanonicalForm}, we can assume that $|\Phi\>$ has the form $|\Phi\> = \cos(\xi/2)|00\> + \sin(\xi/2)|11\>$ with $\xi \in [0,\pi/2]$ without loss of generality. According to \rcite{BagaMM05}, we can further assume that $\scrM$ is a rank-1 projective measurement that has the form $\scrM = \{|\varphi_+\>\<\varphi_+|,|\varphi_-\>\<\varphi_-|\}$, where 
	\begin{equation}
		|\varphi_+\> = \cos\frac{\theta}{2}|0\> + \sin\frac{\theta}{2} \rme^{\rmi \phi}|1\>,\quad |\varphi_-\> = \sin\frac{\theta}{2}|0\> - \cos\frac{\theta}{2} \rme^{\rmi \phi}|1\>,\quad \theta\in[0,\pi],\; \phi\in[0,2\pi). 
	\end{equation}
	Then
	\begin{align}
		\sum_j \|\caQ(|\Phi\>\<\Phi|\otimes M_j)\| & = \|\caQ(|\Phi\>\<\Phi|\otimes|\varphi_+\>\<\varphi_+|)\| + \|\caQ(|\Phi\>\<\Phi|\otimes|\varphi_-\>\<\varphi_-|)\| \nonumber\\
		&= 20 + \sqrt{\sin^2\theta(9 + \sin^2\xi + 6\sin\xi \cos 2\phi) + (5\cos\xi + 4\cos\theta)^2}\nonumber \\
		&\equad + \sqrt{\sin^2\theta(9 + \sin^2\xi + 6\sin\xi \cos 2\phi) + (5\cos\xi - 4\cos\theta)^2}\nonumber \\
		&\le 20+q(\xi,\theta),  \label{eq:PsiPOVM2+1Proof}
	\end{align}
	where 
	\begin{equation}
		q(\xi,\theta) := \sqrt{h_+}+\sqrt{h_-},\quad h_\pm := \sin^2\theta(3 + \sin\xi)^2 + (5\cos\xi \pm 4\cos\theta)^2,
	\end{equation}
	and the inequality above is saturated when $\phi = 0$. To prove \eref{eq:PsiPOVM2+1}, it suffices to prove the following inequality:
	\begin{align}\label{eq:qxithetaUB}
		q(\xi,\theta)\leq\frac{5\sqrt{22}}{2}.
	\end{align}
	If $\theta=0$ or $\theta=\pi$, then 
	\begin{align}
		q(\xi,\theta) &=  |5 \cos \xi + 4| + |5 \cos \xi - 4| \le 10<\frac{5\sqrt{22}}{2}, \label{eq:qxithetaUB1}
	\end{align}
	which confirms \eref{eq:qxithetaUB}.

	Next, suppose $0<\theta<\pi$, then $\sin\theta>0$ and  $h_\pm>0$. 
	To determine the extremal points of $q(\xi,\theta)$, we can evaluate the partial derivative of  $q(\xi,\theta)$ over $\theta$, with the result
	\begin{gather}
		\frac{\partial q(\xi,\theta)}{\partial \theta} = \sin\theta\left(-\frac{g_+}{\sqrt{h_+}} + \frac{g_-}{\sqrt{h_-}}\right) = \frac{\sin\theta \bigl(g_- \sqrt{h_+} - g_+ \sqrt{h_-}\lsp\bigr)}{\sqrt{h_+ h_-}},\\ 
		g_\pm := 20\cos\xi \mp \cos\theta(\sin^2\xi + 6\sin\xi - 7) .
	\end{gather}
	In addition,
	\begin{equation}
		g^2_-h_+ - g^2_+ h_- = 480\cos\theta \left(\cos\frac{\xi}{2} - \sin\frac{\xi}{2}\right)^3\left(\cos\frac{\xi}{2} + \sin\frac{\xi}{2}\right)(3+\sin\xi)^2 (3+4\sin\xi).
	\end{equation}
	If $\partial q(\xi,\theta)/\partial \theta = 0$, then $g^2_-h_+ - g^2_+ h_-=0$, which means $\cos\theta = 0$ or $\cos(\xi/2)=\sin(\xi/2)$, that is, $\theta=\pi/2$ or  $\xi = \pi/2$,  given that $0<\theta<\pi$ and $0\leq \xi\leq \pi/2$ by assumption. In the latter case, we have
	\begin{equation}\label{eq:qxithetaUB2}
		q(\xi,\theta)= q(\pi/2, \theta) = 8 \quad \forall \, \theta. 
	\end{equation}
	In the former case, we have
	\begin{align}\label{eq:qxithetaUB3}
		q(\xi,\theta) = q(\xi,\pi/2) &= 2\sqrt{(3+\sin\xi)^2 + 25\cos^2\xi} \le  \frac{5\sqrt{22}}{2},
	\end{align}
	where the inequality is saturated when $\xi =\xi_0= \arcsin(1/8)$. In conjunction with \eqsref{eq:qxithetaUB1}{eq:qxithetaUB2}, this observation completes the proof of \eref{eq:qxithetaUB}. 
	
	Now, \eref{eq:PsiPOVM2+1} is a simple corollary of \eqsref{eq:PsiPOVM2+1Proof}{eq:qxithetaUB}. If $|\Phi\>$ and $\scrM$ have the form in \eref{eq:PsiPOVM2+1Saturate}, then the inequalities in \eqsref{eq:PsiPOVM2+1Proof}{eq:qxithetaUB3}
	are saturated, so 
	the upper bound in \eref{eq:PsiPOVM2+1} is  saturated accordingly, which can also be verified by straightforward calculation. 
\end{proof}

\subsection{Direct proof of \eref{eq:2+1}}

Thanks to \lref{lem:POVMsym}, to determine the maximum estimation fidelity of $2+1$ adaptive POVMs on $\caH^{\otimes 3}$, it suffices to consider $2+1$ adaptive rank-1 POVMs on $\Sym_2(\caH)\otimes \caH$.

A general  $2+1$ adaptive rank-1 POVM $\scrM$ on $\Sym_2(\caH)\otimes \caH$ can be expressed as follows:
\begin{align}
	\scrM=\bigsqcup_j w_j|\Phi_j\>\<\Phi_j|\otimes \scrM_j, 
\end{align}
where $\{w_j|\Phi_j\>\<\Phi_j|\}_j$ forms a POVM on $\Sym_2(\caH)$, which means $|\Phi_j\>\in \Sym_2(\caH)$,  $w_j> 0$, and $\sum_j w_j=3$; in addition, each $\scrM_j$ is a POVM on $\caH$. 
By virtue of \eref{eq:EstimationFidDef} and \lref{lem:proofofthm2+1} we can deduce that
\begin{align}
	F(\scrM) &= \frac{1}{120}\sum_{j}\sum_{M\in \scrM_j} \|\caQ(w_j |\Phi_j\>\<\Phi_j|\otimes M)\|
	\le \frac{1}{120}\left(20 + \frac{5}{2}\sqrt{22}\right)\sum_{j}w_j =F_{2\rightarrow 1}= \frac{1}{2} + \frac{\sqrt{22}}{16},
\end{align}
which confirms \eref{eq:2+1}. Incidentally,
an optimal $2+1$ adaptive POVM that can attain the maximum estimation fidelity $F_{2\rightarrow 1}$  is presented in \eref{eq:OptPOVMbs2+1} in the main text.

\section{Proof of \thref{thm:1+2}}
\label{app:1+2}
\setcounter{lemma}{0} 

\subsection{Auxiliary results}
As a complement to \lref{lem:fxyUB},  here we first derive another tight linear upper bound for the function $f(x,y)$ defined in \eref{eq:fxy}. Let
\begin{equation}\label{eq:p0gamma}
	\begin{gathered}
		p := \frac{47 - 3\sqrt{41}}{216},\quad x_0 := \frac{\frac{2}{3} - p}{1 - p}=\frac{2003 + 27 \sqrt{41}}{3524},\quad y_0 := \frac{-p}{1 - p}=\frac{-1039 + 81 \sqrt{41}}{3524},\\
		\alpha := f_x(x_0,y_0) = \frac{5}{2} - \frac{43}{2\sqrt{41}},\;\; \beta := f_y(x_0,y_0) = \frac{9}{2} - \frac{13}{2\sqrt{41}},\;\;
		\gamma := f(x_0, y_0) - \alpha x_0 - \beta y_0 = 2 + \frac{28}{\sqrt{41}}.
	\end{gathered}
\end{equation}
Here $p$ is reproduced from \eref{eq:pSPhiW}, $f_x=\partial f/\partial x$, and $f_y=\partial f/\partial y$. Note that $\alpha$ and $y_0$ are negative, while the other four numbers are positive. 
\begin{lemma}\label{lem:fxyABC}
	Suppose  $0 \le x \le 1$ and $-x \le y \le x$. Then the function $f(x,y)$ defined in \eref{eq:fxy} satisfies
	\begin{equation}\label{eq:fxyABC}
		f(x,y) \le \alpha x + \beta y + \gamma,
	\end{equation}
	where   $\alpha,\beta,\gamma$ are defined in \eref{eq:p0gamma}, and   the inequality is saturated iff $x = y = 1$ or $x = x_0,  y = y_0$.
\end{lemma}

\begin{proof}
	Note that $\alpha x + \beta y + \gamma\geq \gamma+\alpha-\beta=13/\sqrt{41}>0$. 
	Define the difference function
	\begin{align}
		\Delta(x,y) &:= (\alpha x + \beta y + \gamma)^2 - f(x,y)^2\nonumber\\
		&\;= -4 - 108 x + 96 x^2 - 40 y - 
		25 y^2 + 48 (-1 + x)\bigl(-x + \sqrt{x^2 - y^2}\lsp\bigr)\nonumber\\
		&\;\equad + \left[ \left(\frac{5}{2} - \frac{43}{2\sqrt{41}}\right)x + \left(\frac{9}{2} - \frac{13}{2\sqrt{41}}\right)y + 2 + \frac{28}{\sqrt{41}}\right]^2.
	\end{align}
	Then, to prove \eref{eq:fxyABC}, it suffices to prove the inequality  $\Delta(x,y) \ge 0$ for  $0 \le x \le 1$ and $-x \le y \le x$. When $x=0$, which means $y=0$, it is straightforward to verify that $\Delta(x,y)>0$ and  $f(x,y)<\alpha x + \beta y + \gamma$.

	By assumption $y$ can be expressed as  $y = x\cos \zeta$ with $\zeta\in [0,\pi]$, and $\zeta$ is uniquely determined by $x$ and $y$ when $x\neq 0$. Accordingly, $\Delta(x,y)$ can be expressed as 
	\begin{gather}
		\Delta(x,y) =\Delta(x,x\cos \zeta)= 
		g_2(\zeta) x^2-2g_1(\zeta)x  + \gamma^2 - 4, 
	\end{gather}
	where
\begin{equation}\label{eq:g12}
	g_2(\zeta): =48 + (\alpha + \beta\cos \zeta)^2 - 25\cos^2 \zeta + 48\sin \zeta ,\quad    g_1(\zeta):=30 - \alpha\gamma - (\beta\gamma - 20)\cos \zeta + 24\sin \zeta.
\end{equation}
Note that	
\begin{equation}\label{eq:g12LB}
g_2(\zeta) \ge 23,\quad      g_1(\zeta)\geq g_1(0)=50-\alpha\gamma - \beta\gamma>33,
 \end{equation}
given that $\beta\gamma-20>0$. Let $x^*(\zeta) : =g_1(\zeta)/g_2(\zeta)$; then  $\Delta(x,x\cos \zeta)\geq \Delta(x^*(\zeta),x^*(\zeta)\cos \zeta)$, and the inequality is saturated iff $x= x^*(\zeta)\leq 1$. 
	
Let 
\begin{equation}\label{eq:ellzeta}
	\ell(\zeta) :=g_2(\zeta)-g_1(\zeta)= 18 + \alpha^2 + \alpha \gamma - (20 - \beta\gamma - 2\alpha\beta)\cos \zeta  -\bigl(25 - \beta^2\bigr)\cos^2 \zeta + 24\sin \zeta; 
\end{equation}
then $x^*(\zeta)\leq 1$ iff $\ell(\zeta) \geq 0$.  If $\pi/2\leq \zeta\leq \pi$, then 
\begin{align}
	\ell(\zeta)\geq 18 + \alpha^2 + \alpha \gamma-\bigl(25 - \beta^2\bigr)=\frac{907-139\sqrt{41}}{41}>0,\quad x^*(\zeta)<1, 
\end{align}
given that $(20 - \beta\gamma - 2\alpha\beta)>0$ and $25 - \beta^2>0$.
If instead $0\leq \zeta\leq \pi/2$, then $\ell(\zeta)$ is monotonically increasing in $\zeta$. Meanwhile, $\ell(0)<0$ and $\ell(\pi/2)>0$. Therefore, $\ell(\zeta)$ has a unique zero for $\zeta\in [0,\pi]$. Let $\zeta_*$ be this unique zero; numerically, we have $\zeta_* \approx 0.12988$ and $\cos \zeta_* \approx 0.99158$.  Then $x^*(\zeta_*)=1$,  $x^*(\zeta)>1$ for $\zeta\in [0,\zeta_*)$, and  $x^*(\zeta)<1$ for $\zeta\in (\zeta_*, \pi]$. Define
\begin{equation}
	\Delta^*(\zeta) := \begin{cases}
		\Delta(1,\cos \zeta) & \zeta \in [0,\zeta_*),\\
		\Delta(x^*(\zeta),x^*(\zeta)\cos \zeta) & \zeta \in [\zeta_*,\pi];
	\end{cases}
\end{equation}
then $\Delta(x,x\cos \zeta) \ge \Delta^*(\zeta)$, and the inequality is saturated iff 
\begin{align}\label{eq:xzetaMin}
	x=\begin{cases}
		1 & \zeta \in [0,\zeta_*),\\
		x^*(\zeta) & \zeta \in [\zeta_*,\pi].
	\end{cases}
\end{align}  
To prove \eref{eq:fxyABC}, it suffices to prove the inequality  $\Delta^*(\zeta) \ge 0$ for $\zeta\in [0,\pi]$.

If $\zeta \in [0,\zeta_*)$, then 
\begin{equation}\label{eq:Delta*zetaLB1}
	\Delta^*(\zeta)=\Delta(1,\cos \zeta)=\frac{1}{41}\bigl[433 + 117 \sqrt{41} + \bigl(305 + 117  \sqrt{41}\lsp\bigr) \cos\zeta\bigr]\sin^2\frac{\zeta}{2}\geq 0. 
\end{equation}

If instead $\zeta \in [\zeta_*,\pi]$, then 
\begin{align}
	\Delta^*(\zeta)=\Delta(x^*(\zeta),x^*(\zeta)\cos \zeta)=\frac{(\gamma^2-4)g_2(\zeta) -g_1(\zeta)^2}{g_2(\zeta)}=\frac{h(\zeta)}{41g_2(\zeta)},
\end{align}
where 
\begin{align}
	h(\zeta) &:= c_0 \cos (2\zeta) + c_1 \sin (2\zeta) + c_2 \cos \zeta + 
	c_3 \sin \zeta + c_4\nonumber\\
	&\;=2c_0\cos^2\zeta+ c_2 \cos \zeta+c_4-c_0+(2c_1\cos\zeta+c_3)\sin\zeta, \label{eq:hzeta}
\end{align}
with
\begin{equation}
	\begin{gathered}
		c_0 = -4197 + 977 \sqrt{41},\quad c_1 = 24 \bigl(-633 + 113 \sqrt{41}\lsp\bigr),\quad c_2 = 4\bigl(-14667 + 2191 \sqrt{41}\lsp\bigr),\\
		c_3 = 48 \bigl(-843 + 139 \sqrt{41}\lsp\bigr),\quad c_4 = -53775 + 8403 \sqrt{41}.
	\end{gathered}
\end{equation}
Note that $c_0,c_1, c_3,c_4>0$ and $c_2<0$; in addition, $g_2(\zeta)>0$ by \eref{eq:g12LB}. To prove the inequality $\Delta^*(\zeta)\geq 0$ for  $\zeta \in [\zeta_*,\pi]$, it suffices to prove the inequality $h(\zeta)\geq 0$.

Let $u=\cos\zeta$ and define
\begin{equation}
	h_2(u) := \bigl(2c_0 u^2 + c_2 u + c_4 - c_0\bigr)^2 - \left(2c_1 u + c_3\right)^2\bigl(1-u^2\bigr).
\end{equation}
Then $h_2(u)=0$ whenever $h(\zeta)=0$ according to \eref{eq:hzeta}. Calculation shows that $h_2(u)$ has the following three distinct zeros:
\begin{equation}\label{eq:u0pm}
	u_0 := \frac{-308 + 27\sqrt{41}}{565},\quad 
	u_\pm:=\frac{37529139\sqrt{41}-239145719\pm 576\sqrt{-35743460158+5587351798\sqrt{41}}}{294550033-45301173\sqrt{41}},
\end{equation}
where $u_-<u_0<u_+$ and 
the zero $u_0$ has multiplicity 2.
By contrast, $h(\zeta)$ has two distinct zeros, namely,  $\zeta_0:=\arccos u_0\approx 1.81228$ and $\zeta_+:=\arccos u_+\approx 0.07235$; note that $\arccos u_-$ is not a zero of $h(\zeta)$. 
In addition, $0<\zeta_+<\zeta_*$ and $\zeta_*<\zeta_0<\pi$, so $\zeta_0$ is the unique zero of $h(\zeta)$ within the interval $[\zeta_*,\pi]$. Straightforward calculation shows that $h(\zeta_*), h(\pi)>0$ as illustrated in \fref{fig:hzeta}, which implies that $h(\zeta),\Delta^*(\zeta)  \geq 0$ for $\zeta\in [\zeta_*,\pi]$ by continuity. 
In conjunction with \eref{eq:Delta*zetaLB1} we can deduce   that $\Delta(x,x\cos \zeta)\geq \Delta^*(\zeta)\geq 0$ for $\zeta\in [0,\pi]$, which implies  \eref{eq:fxyABC}; in addition, $\Delta^*(\zeta)$ has only two zeros in this interval, namely, 0 and $\zeta_0$.

If  $x = y = 1$ or if  $x = x_0$ and $ y = y_0$,
then the inequality in \eref{eq:fxyABC} is saturated by straightforward calculation. Conversely, if the inequality in \eref{eq:fxyABC} is saturated, then $\Delta^*(\zeta)=\Delta(x,x\cos \zeta)=\Delta(x,y)=0$, which means $\zeta=0$ or $\zeta=\zeta_0$. According to \eref{eq:xzetaMin}, if $\zeta=0$, then $y=x=1$; if instead $\zeta=\zeta_0$, then 
\begin{align}
	x=x^*(\zeta_0)=x_0, \quad y=x^*(\zeta_0)\cos\zeta_0=y_0.  
\end{align}
This observation completes the proof of \lref{lem:fxyABC}.
\end{proof}

\begin{figure}
	\centering
	\includegraphics[width = 9cm]{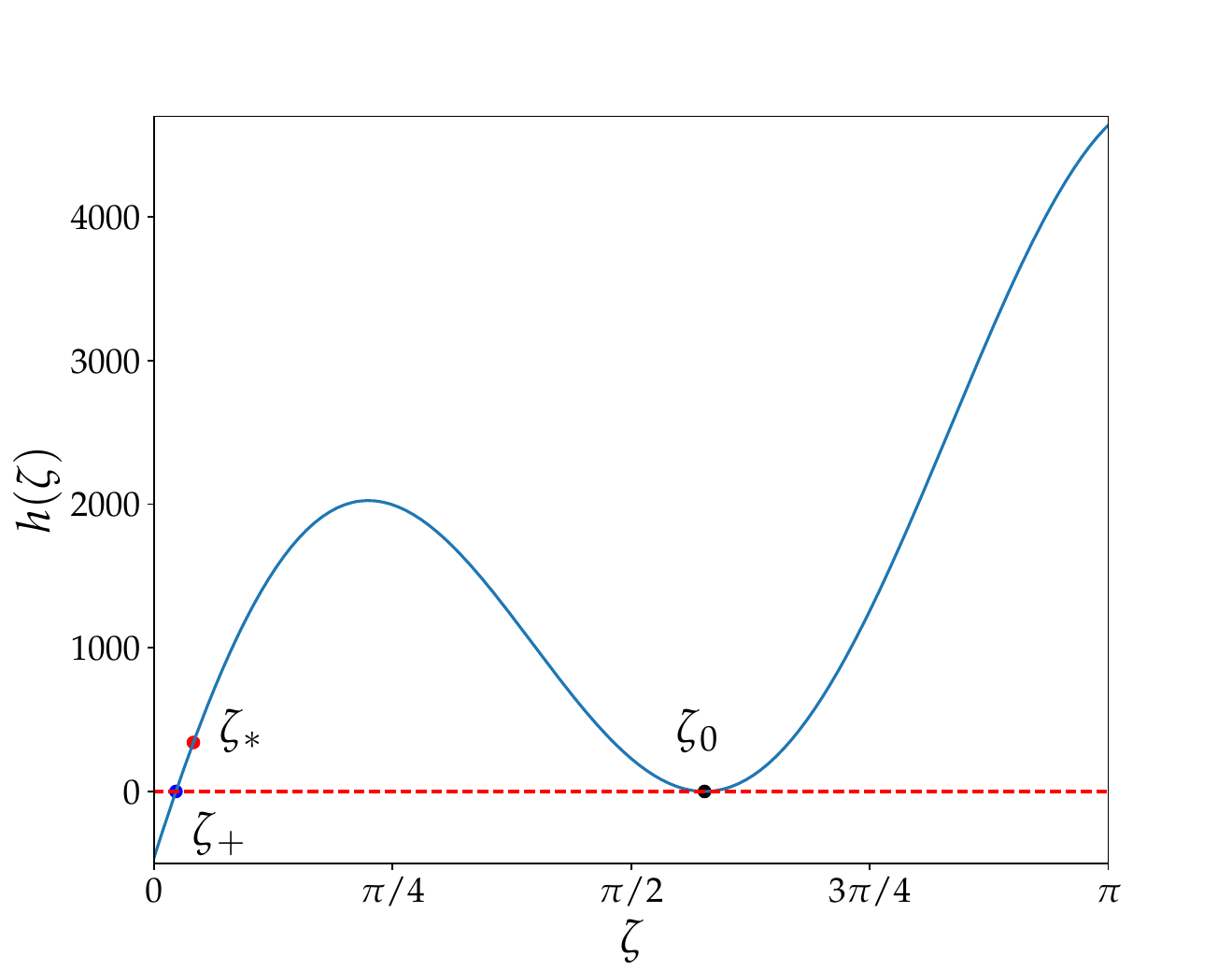}
	\caption{\label{fig:hzeta}A plot  of the function $h(\zeta)$ defined in \eref{eq:hzeta} for $\zeta\in [0,\pi]$. Here $\zeta_+=\arccos u_+$ and $\zeta_0=\arccos u_0$ are the two zeros of  $h(\zeta)$ [see  \eref{eq:u0pm}], while $\zeta_*$ is the unique zero of the function $\ell(\zeta)$ defined in \eref{eq:ellzeta}.}
\end{figure}

\begin{lemma}\label{lem:canonical1+2}
	Suppose $\scrM = \{w_j |\Psi_j\>\<\Psi_j|\}_j$ is a $1+2$ adaptive POVM on $\Sym_2(\caH)\otimes\caH$, where $|\Psi_j\>$ have the form
	\begin{equation}
		|\Psi_j\> = U_j^{\otimes 3}\bigl[\bigl(a_j|00\> + b_j|\bbS\> + c_j \rme^{\rmi \phi_j}|11\>\bigr)\otimes|0\>\bigr],
	\end{equation}
	with $U_j\in \rmU(\caH), a_j, b_j, c_j \ge 0, a_j^2 + b_j^2 + c_j^2 = 1 $, and $\phi_j\in[0,2\pi)$. Then 
	\begin{equation}\label{eq:constraint2}
		\sum_j w_j a_j^2 = \sum_j w_j b_j^2 = \sum_j w_j c_j^2 = 2,\quad \sum_j w_j a_j b_j = \sum_j a_j c_j \rme^{\rmi\phi_j} = \sum_j b_j c_j \rme^{\rmi\phi_j} = 0.
	\end{equation}
\end{lemma}

\begin{proof}
Let $|\Phi_j\>=a_j|00\> + b_j|\bbS\> + c_j \rme^{\rmi \phi_j}|11\>$, then the set $\{(w_j/2) |\Phi_j\>\<\Phi_j|\}_j$ forms a POVM on $\Sym_2(\caH)$. Therefore,
\begin{align}
\sum_j w_j|\Phi_j\>\<\Phi_j|=2P_2 =2|00\>\<00|+2|\bbS\>\<\bbS|+2|11\>\<11|, 
\end{align}
which implies \eref{eq:constraint2} and completes the proof of \lref{lem:canonical1+2}. 
\end{proof}

\subsection{Proof of \thref{thm:1+2}}

Thanks to \lref{lem:POVMsym}, to determine the maximum estimation fidelity of $1+2$ adaptive POVMs on $\caH^{\otimes 3}$, it suffices to consider $1+2$ adaptive rank-1 POVMs on $\Sym_2(\caH)\otimes \caH$.

Suppose $\scrM = \bigl\{M_j\}_j$
is an arbitrary $1+2$ rank-1 POVM on $\Sym_2(\caH)\otimes\caH$. Then the POVM elements $M_j$ can be expressed as $M_j=w_j U_j^{\otimes 3}|\Psi_j\>\<\Psi_j|U_j^{\dag\otimes 3}$, where
$w_j> 0$, $\sum_j w_j=6$, $U_j\in \rmU(\caH)$, and $|\Psi_j\>$ have the form
\begin{equation}
	|\Psi_j\> = \bigl(a_j|00\> + b_j|\bbS\> + c_j \rme^{\rmi \phi_j}|11\>\bigr)\otimes|0\>
\end{equation}
with $a_j,b_j,c_j\ge 0, a_j^2 + b_j^2 + c_j^2 = 1$, and $\phi_j\in[0,2\pi)$. By virtue of \eref{eq:EstimationFidDef} and \lsref{lem:q_abc} and \ref{lem:canonical1+2}, we can deduce that
\begin{align}
    F(\scrM)&=\frac{1}{120}\sum_j w_j \caQ(|\Psi_j\>\<\Psi_j|)=\frac{1}{120}\sum_j w_j \eta(a_j,b_j,c_j,\phi_j)
    \nonumber\\
    &\le \frac{1}{120}\sum_j w_j\bigl[15 a_j^2+10b_j^2+5c_j^2 +f(x_j,y_j)\bigr]
    =\frac{1}{2}+\frac{1}{20}\sum_{j} p_jf(x_j,y_j),  \label{eq:1+2sat}
\end{align}
where the relevant parameters satisfy the following constraints (see \lref{lem:canonical1+2}):
\begin{equation}\label{eq:opt3}
 p_j := \frac{w_j}{6},\quad x_j := a_j^2 + c_j^2,\quad y_j := a_j^2 - c_j^2,\quad
     \sum_{j} p_j = 1, \quad 
     \sum_{j} p_j y_j = 0,\quad \sum_j p_j x_j = \frac{2}{3}.
\end{equation}
Note that the inequality in \eref{eq:1+2sat} is saturated when $\phi_j=0$ for all $j$.

Next, in conjunction with \eref{eq:p0gamma} and \lref{lem:fxyABC} we can deduce that 	
\begin{equation}\label{eq:proofthe4}
F(\scrM)
\le \frac{1}{2}+\frac{1}{20}\sum_j p_j \left(\alpha x_j + \beta y_j + \gamma\right)
= \frac{1}{2} + \frac{1}{30}\alpha + \frac{1}{20}\gamma
= \frac{1}{2} + \frac{11 + \sqrt{41}}{60}.
\end{equation}
The saturation of the above inequality means either $(x_j,y_j)=(x_0,y_0)$  or $(x_j,y_j)=(1,1)$ for each $j$, where $x_0$ and $y_0$ are defined in \eref{eq:p0gamma}. In conjunction with \eref{eq:opt3} we can deduce that
\begin{align}
\sum_{j\,|\,(x_j,y_j)=(1,1)} p_j=p,
\end{align}
where $p$ is defined in \eref{eq:pSPhiW} and is reproduced in \eref{eq:p0gamma}. Moreover, the upper bound in \eref{eq:proofthe4} is saturated when $\scrM$ has the form 
\begin{equation}
\scrM=\{\scrK_0'\otimes |0\>\<0|\}\cup \{\scrK_1'\otimes |1\>\<1|\},\quad 		\scrK_0' := \{K_j\}_{j=0}^4,\quad \scrK_1':=\bigl\{X^{\otimes 2}K_jX^{\otimes 2}\bigr\}_{j=0}^4,
\end{equation}
where $K_j$ for $j=0,1,2,3,4$ are defined in \eref{eq:scrE01}. Note that $\scrK_0'$ and $\scrK_1'$ are POVMs on $\Sym_2(\caH)$. Accordingly, we can construct an optimal $1+2$ adaptive POVM on $\caH^{\otimes 3}$ as shown in \eref{eq:OptPOVM1+2} in the main text. This observation completes the proof of \thref{thm:1+2}. 

\bibliography{all_references}

\begin{thebibliography}{54}%
\makeatletter
\providecommand \@ifxundefined [1]{%
 \@ifx{#1\undefined}
}%
\providecommand \@ifnum [1]{%
 \ifnum #1\expandafter \@firstoftwo
 \else \expandafter \@secondoftwo
 \fi
}%
\providecommand \@ifx [1]{%
 \ifx #1\expandafter \@firstoftwo
 \else \expandafter \@secondoftwo
 \fi
}%
\providecommand \natexlab [1]{#1}%
\providecommand \enquote  [1]{``#1''}%
\providecommand \bibnamefont  [1]{#1}%
\providecommand \bibfnamefont [1]{#1}%
\providecommand \citenamefont [1]{#1}%
\providecommand \href@noop [0]{\@secondoftwo}%
\providecommand \href [0]{\begingroup \@sanitize@url \@href}%
\providecommand \@href[1]{\@@startlink{#1}\@@href}%
\providecommand \@@href[1]{\endgroup#1\@@endlink}%
\providecommand \@sanitize@url [0]{\catcode `\\12\catcode `\$12\catcode
  `\&12\catcode `\#12\catcode `\^12\catcode `\_12\catcode `\%12\relax}%
\providecommand \@@startlink[1]{}%
\providecommand \@@endlink[0]{}%
\providecommand \url  [0]{\begingroup\@sanitize@url \@url }%
\providecommand \@url [1]{\endgroup\@href {#1}{\urlprefix }}%
\providecommand \urlprefix  [0]{URL }%
\providecommand \Eprint [0]{\href }%
\providecommand \doibase [0]{https://doi.org/}%
\providecommand \selectlanguage [0]{\@gobble}%
\providecommand \bibinfo  [0]{\@secondoftwo}%
\providecommand \bibfield  [0]{\@secondoftwo}%
\providecommand \translation [1]{[#1]}%
\providecommand \BibitemOpen [0]{}%
\providecommand \bibitemStop [0]{}%
\providecommand \bibitemNoStop [0]{.\EOS\space}%
\providecommand \EOS [0]{\spacefactor3000\relax}%
\providecommand \BibitemShut  [1]{\csname bibitem#1\endcsname}%
\let\auto@bib@innerbib\@empty
\bibitem [{\citenamefont {Nielsen}\ and\ \citenamefont
  {Chuang}(2010)}]{NielC10book}%
  \BibitemOpen
  \bibfield  {author} {\bibinfo {author} {\bibfnamefont {M.~A.}\ \bibnamefont
  {Nielsen}}\ and\ \bibinfo {author} {\bibfnamefont {I.~L.}\ \bibnamefont
  {Chuang}},\ }\href@noop {} {\emph {\bibinfo {title} {Quantum Computation and
  Quantum Information}}},\ \bibinfo {edition} {2nd}\ ed.\ (\bibinfo
  {publisher} {Cambridge University Press},\ \bibinfo {address} {Cambridge,
  UK},\ \bibinfo {year} {2010})\BibitemShut {NoStop}%
\bibitem [{\citenamefont {Peres}\ and\ \citenamefont
  {Wootters}(1991)}]{PereW91}%
  \BibitemOpen
  \bibfield  {author} {\bibinfo {author} {\bibfnamefont {A.}~\bibnamefont
  {Peres}}\ and\ \bibinfo {author} {\bibfnamefont {W.~K.}\ \bibnamefont
  {Wootters}},\ }\bibfield  {title} {\bibinfo {title} {Optimal detection of
  quantum information},\ }\href@noop {} {\bibfield  {journal} {\bibinfo
  {journal} {Phys. Rev. Lett.}\ }\textbf {\bibinfo {volume} {66}},\ \bibinfo
  {pages} {1119} (\bibinfo {year} {1991})}\BibitemShut {NoStop}%
\bibitem [{\citenamefont {Massar}\ and\ \citenamefont
  {Popescu}(1995)}]{MassP95}%
  \BibitemOpen
  \bibfield  {author} {\bibinfo {author} {\bibfnamefont {S.}~\bibnamefont
  {Massar}}\ and\ \bibinfo {author} {\bibfnamefont {S.}~\bibnamefont
  {Popescu}},\ }\bibfield  {title} {\bibinfo {title} {Optimal extraction of
  information from finite quantum ensembles},\ }\href@noop {} {\bibfield
  {journal} {\bibinfo  {journal} {Phys. Rev. Lett.}\ }\textbf {\bibinfo
  {volume} {74}},\ \bibinfo {pages} {1259} (\bibinfo {year}
  {1995})}\BibitemShut {NoStop}%
\bibitem [{\citenamefont {Gisin}\ and\ \citenamefont
  {Popescu}(1999)}]{GisiP99}%
  \BibitemOpen
  \bibfield  {author} {\bibinfo {author} {\bibfnamefont {N.}~\bibnamefont
  {Gisin}}\ and\ \bibinfo {author} {\bibfnamefont {S.}~\bibnamefont
  {Popescu}},\ }\bibfield  {title} {\bibinfo {title} {Spin flips and quantum
  information for antiparallel spins},\ }\href@noop {} {\bibfield  {journal}
  {\bibinfo  {journal} {Phys. Rev. Lett.}\ }\textbf {\bibinfo {volume} {83}},\
  \bibinfo {pages} {432} (\bibinfo {year} {1999})}\BibitemShut {NoStop}%
\bibitem [{\citenamefont {Massar}(2000)}]{Mass00}%
  \BibitemOpen
  \bibfield  {author} {\bibinfo {author} {\bibfnamefont {S.}~\bibnamefont
  {Massar}},\ }\bibfield  {title} {\bibinfo {title} {Collective versus local
  measurements on two parallel or antiparallel spins},\ }\href@noop {}
  {\bibfield  {journal} {\bibinfo  {journal} {Phys. Rev. A}\ }\textbf {\bibinfo
  {volume} {62}},\ \bibinfo {pages} {040101(R)} (\bibinfo {year}
  {2000})}\BibitemShut {NoStop}%
\bibitem [{\citenamefont {Bagan}\ \emph {et~al.}(2006)\citenamefont {Bagan},
  \citenamefont {Ballester}, \citenamefont {Gill}, \citenamefont
  {{Mu\~noz-Tapia}},\ and\ \citenamefont {Romero-Isart}}]{BagaBGM06S}%
  \BibitemOpen
  \bibfield  {author} {\bibinfo {author} {\bibfnamefont {E.}~\bibnamefont
  {Bagan}}, \bibinfo {author} {\bibfnamefont {M.~A.}\ \bibnamefont
  {Ballester}}, \bibinfo {author} {\bibfnamefont {R.~D.}\ \bibnamefont {Gill}},
  \bibinfo {author} {\bibfnamefont {R.}~\bibnamefont {{Mu\~noz-Tapia}}},\ and\
  \bibinfo {author} {\bibfnamefont {O.}~\bibnamefont {Romero-Isart}},\
  }\bibfield  {title} {\bibinfo {title} {Separable measurement estimation of
  density matrices and its fidelity gap with collective protocols},\
  }\href@noop {} {\bibfield  {journal} {\bibinfo  {journal} {Phys. Rev. Lett.}\
  }\textbf {\bibinfo {volume} {97}},\ \bibinfo {pages} {130501} (\bibinfo
  {year} {2006})}\BibitemShut {NoStop}%
\bibitem [{\citenamefont {Zhu}(2012)}]{Zhu12the}%
  \BibitemOpen
  \bibfield  {author} {\bibinfo {author} {\bibfnamefont {H.}~\bibnamefont
  {Zhu}},\ }\emph {\bibinfo {title} {Quantum State Estimation and Symmetric
  Informationally Complete {POM}s}},\ \href@noop {} {Ph.D. thesis},\ \bibinfo
  {school} {National University of Singapore} (\bibinfo {year}
  {2012})\BibitemShut {NoStop}%
\bibitem [{\citenamefont {Zhu}\ and\ \citenamefont {Hayashi}(2018)}]{ZhuH18U}%
  \BibitemOpen
  \bibfield  {author} {\bibinfo {author} {\bibfnamefont {H.}~\bibnamefont
  {Zhu}}\ and\ \bibinfo {author} {\bibfnamefont {M.}~\bibnamefont {Hayashi}},\
  }\bibfield  {title} {\bibinfo {title} {Universally {Fisher}-symmetric
  informationally complete measurements},\ }\href@noop {} {\bibfield  {journal}
  {\bibinfo  {journal} {Phys. Rev. Lett.}\ }\textbf {\bibinfo {volume} {120}},\
  \bibinfo {pages} {030404} (\bibinfo {year} {2018})}\BibitemShut {NoStop}%
\bibitem [{\citenamefont {Haah}\ \emph {et~al.}(2016)\citenamefont {Haah},
  \citenamefont {Harrow}, \citenamefont {Ji}, \citenamefont {Wu},\ and\
  \citenamefont {Yu}}]{HaahHJW16}%
  \BibitemOpen
  \bibfield  {author} {\bibinfo {author} {\bibfnamefont {J.}~\bibnamefont
  {Haah}}, \bibinfo {author} {\bibfnamefont {A.~W.}\ \bibnamefont {Harrow}},
  \bibinfo {author} {\bibfnamefont {Z.}~\bibnamefont {Ji}}, \bibinfo {author}
  {\bibfnamefont {X.}~\bibnamefont {Wu}},\ and\ \bibinfo {author}
  {\bibfnamefont {N.}~\bibnamefont {Yu}},\ }\bibfield  {title} {\bibinfo
  {title} {Sample-optimal tomography of quantum states},\ }in\ \href@noop {}
  {\emph {\bibinfo {booktitle} {Proceedings of the 48th Annual ACM Symposium on
  Theory of Computing (STOC)}}}\ (\bibinfo {organization} {ACM, New York},\
  \bibinfo {year} {2016})\ pp.\ \bibinfo {pages} {913--925}\BibitemShut
  {NoStop}%
\bibitem [{\citenamefont {O'Donnell}\ and\ \citenamefont
  {Wright}(2016)}]{ODonW16}%
  \BibitemOpen
  \bibfield  {author} {\bibinfo {author} {\bibfnamefont {R.}~\bibnamefont
  {O'Donnell}}\ and\ \bibinfo {author} {\bibfnamefont {J.}~\bibnamefont
  {Wright}},\ }\bibfield  {title} {\bibinfo {title} {Efficient quantum
  tomography},\ }in\ \href@noop {} {\emph {\bibinfo {booktitle} {Proceedings of
  the 48th Annual ACM Symposium on Theory of Computing (STOC)}}}\ (\bibinfo
  {organization} {ACM, New York},\ \bibinfo {year} {2016})\ pp.\ \bibinfo
  {pages} {899--912}\BibitemShut {NoStop}%
\bibitem [{\citenamefont {Vidrighin}\ \emph {et~al.}(2014)\citenamefont
  {Vidrighin}, \citenamefont {Donati}, \citenamefont {Genoni}, \citenamefont
  {Jin}, \citenamefont {Kolthammer}, \citenamefont {Kim}, \citenamefont
  {Datta}, \citenamefont {Barbieri},\ and\ \citenamefont
  {Walmsley}}]{VidrDGJ14}%
  \BibitemOpen
  \bibfield  {author} {\bibinfo {author} {\bibfnamefont {M.~D.}\ \bibnamefont
  {Vidrighin}}, \bibinfo {author} {\bibfnamefont {G.}~\bibnamefont {Donati}},
  \bibinfo {author} {\bibfnamefont {M.~G.}\ \bibnamefont {Genoni}}, \bibinfo
  {author} {\bibfnamefont {X.-M.}\ \bibnamefont {Jin}}, \bibinfo {author}
  {\bibfnamefont {W.~S.}\ \bibnamefont {Kolthammer}}, \bibinfo {author}
  {\bibfnamefont {M.~S.}\ \bibnamefont {Kim}}, \bibinfo {author} {\bibfnamefont
  {A.}~\bibnamefont {Datta}}, \bibinfo {author} {\bibfnamefont
  {M.}~\bibnamefont {Barbieri}},\ and\ \bibinfo {author} {\bibfnamefont
  {I.~A.}\ \bibnamefont {Walmsley}},\ }\bibfield  {title} {\bibinfo {title}
  {Joint estimation of phase and phase diffusion for quantum metrology},\
  }\href@noop {} {\bibfield  {journal} {\bibinfo  {journal} {Nat. Commun.}\
  }\textbf {\bibinfo {volume} {5}},\ \bibinfo {pages} {3532} (\bibinfo {year}
  {2014})}\BibitemShut {NoStop}%
\bibitem [{\citenamefont {Lu}\ and\ \citenamefont {Wang}(2021)}]{LuW21}%
  \BibitemOpen
  \bibfield  {author} {\bibinfo {author} {\bibfnamefont {X.-M.}\ \bibnamefont
  {Lu}}\ and\ \bibinfo {author} {\bibfnamefont {X.}~\bibnamefont {Wang}},\
  }\bibfield  {title} {\bibinfo {title} {Incorporating {Heisenberg's}
  uncertainty principle into quantum multiparameter estimation},\ }\href@noop
  {} {\bibfield  {journal} {\bibinfo  {journal} {Phys. Rev. Lett.}\ }\textbf
  {\bibinfo {volume} {126}},\ \bibinfo {pages} {120503} (\bibinfo {year}
  {2021})}\BibitemShut {NoStop}%
\bibitem [{\citenamefont {Chen}\ \emph
  {et~al.}(2022{\natexlab{a}})\citenamefont {Chen}, \citenamefont {Chen},\ and\
  \citenamefont {Yuan}}]{ChenCY22}%
  \BibitemOpen
  \bibfield  {author} {\bibinfo {author} {\bibfnamefont {H.}~\bibnamefont
  {Chen}}, \bibinfo {author} {\bibfnamefont {Y.}~\bibnamefont {Chen}},\ and\
  \bibinfo {author} {\bibfnamefont {H.}~\bibnamefont {Yuan}},\ }\bibfield
  {title} {\bibinfo {title} {Information geometry under hierarchical quantum
  measurement},\ }\href@noop {} {\bibfield  {journal} {\bibinfo  {journal}
  {Phys. Rev. Lett.}\ }\textbf {\bibinfo {volume} {128}},\ \bibinfo {pages}
  {250502} (\bibinfo {year} {2022}{\natexlab{a}})}\BibitemShut {NoStop}%
\bibitem [{\citenamefont {Aaronson}(2018)}]{Aaro18}%
  \BibitemOpen
  \bibfield  {author} {\bibinfo {author} {\bibfnamefont {S.}~\bibnamefont
  {Aaronson}},\ }\bibfield  {title} {\bibinfo {title} {Shadow tomography of
  quantum states},\ }in\ \href@noop {} {\emph {\bibinfo {booktitle}
  {Proceedings of the 50th Annual ACM Symposium on Theory of Computing
  (STOC)}}}\ (\bibinfo {organization} {ACM, New York},\ \bibinfo {year}
  {2018})\ pp.\ \bibinfo {pages} {325--338}\BibitemShut {NoStop}%
\bibitem [{\citenamefont {Grier}\ \emph {et~al.}(2024)\citenamefont {Grier},
  \citenamefont {Pashayan},\ and\ \citenamefont {Schaeffer}}]{GriePS24}%
  \BibitemOpen
  \bibfield  {author} {\bibinfo {author} {\bibfnamefont {D.}~\bibnamefont
  {Grier}}, \bibinfo {author} {\bibfnamefont {H.}~\bibnamefont {Pashayan}},\
  and\ \bibinfo {author} {\bibfnamefont {L.}~\bibnamefont {Schaeffer}},\
  }\bibfield  {title} {\bibinfo {title} {Sample-optimal classical shadows for
  pure states},\ }\href@noop {} {\bibfield  {journal} {\bibinfo  {journal}
  {Quantum}\ }\textbf {\bibinfo {volume} {8}},\ \bibinfo {pages} {1373}
  (\bibinfo {year} {2024})}\BibitemShut {NoStop}%
\bibitem [{\citenamefont {Liu}\ \emph {et~al.}()\citenamefont {Liu},
  \citenamefont {Li}, \citenamefont {Yuan}, \citenamefont {Zhu},\ and\
  \citenamefont {Zhou}}]{LiuLYZ24}%
  \BibitemOpen
  \bibfield  {author} {\bibinfo {author} {\bibfnamefont {Q.}~\bibnamefont
  {Liu}}, \bibinfo {author} {\bibfnamefont {Z.}~\bibnamefont {Li}}, \bibinfo
  {author} {\bibfnamefont {X.}~\bibnamefont {Yuan}}, \bibinfo {author}
  {\bibfnamefont {H.}~\bibnamefont {Zhu}},\ and\ \bibinfo {author}
  {\bibfnamefont {Y.}~\bibnamefont {Zhou}},\ }\href@noop {} {\bibinfo {title}
  {Auxiliary-free replica shadow estimation}},\ \bibinfo {note}
  {arXiv:2407.20865}\BibitemShut {NoStop}%
\bibitem [{\citenamefont {Li}\ \emph {et~al.}()\citenamefont {Li},
  \citenamefont {Yi}, \citenamefont {Zhou},\ and\ \citenamefont
  {Zhu}}]{LiYZZ24}%
  \BibitemOpen
  \bibfield  {author} {\bibinfo {author} {\bibfnamefont {Z.}~\bibnamefont
  {Li}}, \bibinfo {author} {\bibfnamefont {C.}~\bibnamefont {Yi}}, \bibinfo
  {author} {\bibfnamefont {Y.}~\bibnamefont {Zhou}},\ and\ \bibinfo {author}
  {\bibfnamefont {H.}~\bibnamefont {Zhu}},\ }\bibfield  {title} {\bibinfo
  {title} {Nearly query-optimal classical shadow estimation of unitary
  channels},\ }\href@noop {} {\bibinfo  {journal} {arXiv:2410.14538}\
  }\BibitemShut {NoStop}%
\bibitem [{\citenamefont {Higgins}\ \emph {et~al.}(2011)\citenamefont
  {Higgins}, \citenamefont {Doherty}, \citenamefont {Bartlett}, \citenamefont
  {Pryde},\ and\ \citenamefont {Wiseman}}]{HiggDBP11}%
  \BibitemOpen
\bibfield  {journal} {  }\bibfield  {author} {\bibinfo {author} {\bibfnamefont
  {B.~L.}\ \bibnamefont {Higgins}}, \bibinfo {author} {\bibfnamefont {A.~C.}\
  \bibnamefont {Doherty}}, \bibinfo {author} {\bibfnamefont {S.~D.}\
  \bibnamefont {Bartlett}}, \bibinfo {author} {\bibfnamefont {G.~J.}\
  \bibnamefont {Pryde}},\ and\ \bibinfo {author} {\bibfnamefont {H.~M.}\
  \bibnamefont {Wiseman}},\ }\bibfield  {title} {\bibinfo {title}
  {Multiple-copy state discrimination: Thinking globally, acting locally},\
  }\href@noop {} {\bibfield  {journal} {\bibinfo  {journal} {Phys. Rev. A}\
  }\textbf {\bibinfo {volume} {83}},\ \bibinfo {pages} {052314} (\bibinfo
  {year} {2011})}\BibitemShut {NoStop}%
\bibitem [{\citenamefont {Mart\'{\i}nez~Vargas}\ \emph
  {et~al.}(2021)\citenamefont {Mart\'{\i}nez~Vargas}, \citenamefont {Hirche},
  \citenamefont {Sent\'{\i}s}, \citenamefont {Skotiniotis}, \citenamefont
  {Carrizo}, \citenamefont {Mu\~noz Tapia},\ and\ \citenamefont
  {Calsamiglia}}]{MartHSS21}%
  \BibitemOpen
  \bibfield  {author} {\bibinfo {author} {\bibfnamefont {E.}~\bibnamefont
  {Mart\'{\i}nez~Vargas}}, \bibinfo {author} {\bibfnamefont {C.}~\bibnamefont
  {Hirche}}, \bibinfo {author} {\bibfnamefont {G.}~\bibnamefont {Sent\'{\i}s}},
  \bibinfo {author} {\bibfnamefont {M.}~\bibnamefont {Skotiniotis}}, \bibinfo
  {author} {\bibfnamefont {M.}~\bibnamefont {Carrizo}}, \bibinfo {author}
  {\bibfnamefont {R.}~\bibnamefont {Mu\~noz Tapia}},\ and\ \bibinfo {author}
  {\bibfnamefont {J.}~\bibnamefont {Calsamiglia}},\ }\bibfield  {title}
  {\bibinfo {title} {Quantum sequential hypothesis testing},\ }\href@noop {}
  {\bibfield  {journal} {\bibinfo  {journal} {Phys. Rev. Lett.}\ }\textbf
  {\bibinfo {volume} {126}},\ \bibinfo {pages} {180502} (\bibinfo {year}
  {2021})}\BibitemShut {NoStop}%
\bibitem [{\citenamefont {Huang}\ \emph {et~al.}(2021)\citenamefont {Huang},
  \citenamefont {Kueng},\ and\ \citenamefont {Preskill}}]{HuanKP21}%
  \BibitemOpen
  \bibfield  {author} {\bibinfo {author} {\bibfnamefont {H.-Y.}\ \bibnamefont
  {Huang}}, \bibinfo {author} {\bibfnamefont {R.}~\bibnamefont {Kueng}},\ and\
  \bibinfo {author} {\bibfnamefont {J.}~\bibnamefont {Preskill}},\ }\bibfield
  {title} {\bibinfo {title} {Information-theoretic bounds on quantum advantage
  in machine learning},\ }\href@noop {} {\bibfield  {journal} {\bibinfo
  {journal} {Phys. Rev. Lett.}\ }\textbf {\bibinfo {volume} {126}},\ \bibinfo
  {pages} {190505} (\bibinfo {year} {2021})}\BibitemShut {NoStop}%
\bibitem [{\citenamefont {Chen}\ \emph
  {et~al.}(2022{\natexlab{b}})\citenamefont {Chen}, \citenamefont {Cotler},
  \citenamefont {Huang},\ and\ \citenamefont {Li}}]{ChenCHL22}%
  \BibitemOpen
  \bibfield  {author} {\bibinfo {author} {\bibfnamefont {S.}~\bibnamefont
  {Chen}}, \bibinfo {author} {\bibfnamefont {J.}~\bibnamefont {Cotler}},
  \bibinfo {author} {\bibfnamefont {H.-Y.}\ \bibnamefont {Huang}},\ and\
  \bibinfo {author} {\bibfnamefont {J.}~\bibnamefont {Li}},\ }\bibfield
  {title} {\bibinfo {title} {Exponential separations between learning with and
  without quantum memory},\ }in\ \href@noop {} {\emph {\bibinfo {booktitle}
  {2021 IEEE 62nd Annual Symposium on Foundations of Computer Science
  (FOCS)}}}\ (\bibinfo {organization} {IEEE, New York},\ \bibinfo {year}
  {2022})\ pp.\ \bibinfo {pages} {574--585}\BibitemShut {NoStop}%
\bibitem [{\citenamefont {Aharonov}\ \emph {et~al.}(2022)\citenamefont
  {Aharonov}, \citenamefont {Cotler},\ and\ \citenamefont {Qi}}]{AharCQ22}%
  \BibitemOpen
  \bibfield  {author} {\bibinfo {author} {\bibfnamefont {D.}~\bibnamefont
  {Aharonov}}, \bibinfo {author} {\bibfnamefont {J.}~\bibnamefont {Cotler}},\
  and\ \bibinfo {author} {\bibfnamefont {X.-L.}\ \bibnamefont {Qi}},\
  }\bibfield  {title} {\bibinfo {title} {Quantum algorithmic measurement},\
  }\href@noop {} {\bibfield  {journal} {\bibinfo  {journal} {Nat. Commun.}\
  }\textbf {\bibinfo {volume} {13}},\ \bibinfo {pages} {887} (\bibinfo {year}
  {2022})}\BibitemShut {NoStop}%
\bibitem [{\citenamefont {Grewal}\ \emph {et~al.}(2024)\citenamefont {Grewal},
  \citenamefont {Iyer}, \citenamefont {Kretschmer},\ and\ \citenamefont
  {Liang}}]{GrewIKL24}%
  \BibitemOpen
  \bibfield  {author} {\bibinfo {author} {\bibfnamefont {S.}~\bibnamefont
  {Grewal}}, \bibinfo {author} {\bibfnamefont {V.}~\bibnamefont {Iyer}},
  \bibinfo {author} {\bibfnamefont {W.}~\bibnamefont {Kretschmer}},\ and\
  \bibinfo {author} {\bibfnamefont {D.}~\bibnamefont {Liang}},\ }\bibfield
  {title} {\bibinfo {title} {Improved stabilizer estimation via {Bell}
  difference sampling},\ }in\ \href@noop {} {\emph {\bibinfo {booktitle}
  {Proceedings of the 56th Annual ACM Symposium on Theory of Computing
  (STOC)}}}\ (\bibinfo {organization} {ACM, New York},\ \bibinfo {year}
  {2024})\ pp.\ \bibinfo {pages} {1352--1363}\BibitemShut {NoStop}%
\bibitem [{\citenamefont {Horodecki}(2003)}]{Horo03}%
  \BibitemOpen
  \bibfield  {author} {\bibinfo {author} {\bibfnamefont {P.}~\bibnamefont
  {Horodecki}},\ }\bibfield  {title} {\bibinfo {title} {From limits of quantum
  operations to multicopy entanglement witnesses and state-spectrum
  estimation},\ }\href@noop {} {\bibfield  {journal} {\bibinfo  {journal}
  {Phys. Rev. A}\ }\textbf {\bibinfo {volume} {68}},\ \bibinfo {pages} {052101}
  (\bibinfo {year} {2003})}\BibitemShut {NoStop}%
\bibitem [{\citenamefont {Linden}\ \emph {et~al.}(1998)\citenamefont {Linden},
  \citenamefont {Massar},\ and\ \citenamefont {Popescu}}]{LindMP98}%
  \BibitemOpen
  \bibfield  {author} {\bibinfo {author} {\bibfnamefont {N.}~\bibnamefont
  {Linden}}, \bibinfo {author} {\bibfnamefont {S.}~\bibnamefont {Massar}},\
  and\ \bibinfo {author} {\bibfnamefont {S.}~\bibnamefont {Popescu}},\
  }\bibfield  {title} {\bibinfo {title} {Purifying noisy entanglement requires
  collective measurements},\ }\href@noop {} {\bibfield  {journal} {\bibinfo
  {journal} {Phys. Rev. Lett.}\ }\textbf {\bibinfo {volume} {81}},\ \bibinfo
  {pages} {3279} (\bibinfo {year} {1998})}\BibitemShut {NoStop}%
\bibitem [{\citenamefont {Dehaene}\ \emph {et~al.}(2003)\citenamefont
  {Dehaene}, \citenamefont {Van~den Nest}, \citenamefont {De~Moor},\ and\
  \citenamefont {Verstraete}}]{DehaVDV03}%
  \BibitemOpen
  \bibfield  {author} {\bibinfo {author} {\bibfnamefont {J.}~\bibnamefont
  {Dehaene}}, \bibinfo {author} {\bibfnamefont {M.}~\bibnamefont {Van~den
  Nest}}, \bibinfo {author} {\bibfnamefont {B.}~\bibnamefont {De~Moor}},\ and\
  \bibinfo {author} {\bibfnamefont {F.}~\bibnamefont {Verstraete}},\ }\bibfield
   {title} {\bibinfo {title} {Local permutations of products of {Bell} states
  and entanglement distillation},\ }\href@noop {} {\bibfield  {journal}
  {\bibinfo  {journal} {Phys. Rev. A}\ }\textbf {\bibinfo {volume} {67}},\
  \bibinfo {pages} {022310} (\bibinfo {year} {2003})}\BibitemShut {NoStop}%
\bibitem [{\citenamefont {Eftaxias}\ \emph {et~al.}(2023)\citenamefont
  {Eftaxias}, \citenamefont {Weilenmann},\ and\ \citenamefont
  {Colbeck}}]{EftaWC23}%
  \BibitemOpen
  \bibfield  {author} {\bibinfo {author} {\bibfnamefont {G.}~\bibnamefont
  {Eftaxias}}, \bibinfo {author} {\bibfnamefont {M.}~\bibnamefont
  {Weilenmann}},\ and\ \bibinfo {author} {\bibfnamefont {R.}~\bibnamefont
  {Colbeck}},\ }\bibfield  {title} {\bibinfo {title} {Advantages of multicopy
  nonlocality distillation and its application to minimizing communication
  complexity},\ }\href@noop {} {\bibfield  {journal} {\bibinfo  {journal}
  {Phys. Rev. Lett.}\ }\textbf {\bibinfo {volume} {130}},\ \bibinfo {pages}
  {100201} (\bibinfo {year} {2023})}\BibitemShut {NoStop}%
\bibitem [{\citenamefont {Hou}\ \emph {et~al.}(2018)\citenamefont {Hou},
  \citenamefont {Tang}, \citenamefont {Shang}, \citenamefont {Zhu},
  \citenamefont {Li}, \citenamefont {Yuan}, \citenamefont {Wu}, \citenamefont
  {Xiang}, \citenamefont {Li},\ and\ \citenamefont {Guo}}]{HouTSZ18}%
  \BibitemOpen
  \bibfield  {author} {\bibinfo {author} {\bibfnamefont {Z.}~\bibnamefont
  {Hou}}, \bibinfo {author} {\bibfnamefont {J.-F.}\ \bibnamefont {Tang}},
  \bibinfo {author} {\bibfnamefont {J.}~\bibnamefont {Shang}}, \bibinfo
  {author} {\bibfnamefont {H.}~\bibnamefont {Zhu}}, \bibinfo {author}
  {\bibfnamefont {J.}~\bibnamefont {Li}}, \bibinfo {author} {\bibfnamefont
  {Y.}~\bibnamefont {Yuan}}, \bibinfo {author} {\bibfnamefont {K.-D.}\
  \bibnamefont {Wu}}, \bibinfo {author} {\bibfnamefont {G.-Y.}\ \bibnamefont
  {Xiang}}, \bibinfo {author} {\bibfnamefont {C.-F.}\ \bibnamefont {Li}},\ and\
  \bibinfo {author} {\bibfnamefont {G.-C.}\ \bibnamefont {Guo}},\ }\bibfield
  {title} {\bibinfo {title} {Deterministic realization of collective
  measurements via photonic quantum walks},\ }\href@noop {} {\bibfield
  {journal} {\bibinfo  {journal} {Nat. Commun.}\ }\textbf {\bibinfo {volume}
  {9}},\ \bibinfo {pages} {1414} (\bibinfo {year} {2018})}\BibitemShut
  {NoStop}%
\bibitem [{\citenamefont {Tang}\ \emph {et~al.}(2020)\citenamefont {Tang},
  \citenamefont {Hou}, \citenamefont {Shang}, \citenamefont {Zhu},
  \citenamefont {Xiang}, \citenamefont {Li},\ and\ \citenamefont
  {Guo}}]{TangHSZ20}%
  \BibitemOpen
  \bibfield  {author} {\bibinfo {author} {\bibfnamefont {J.-F.}\ \bibnamefont
  {Tang}}, \bibinfo {author} {\bibfnamefont {Z.}~\bibnamefont {Hou}}, \bibinfo
  {author} {\bibfnamefont {J.}~\bibnamefont {Shang}}, \bibinfo {author}
  {\bibfnamefont {H.}~\bibnamefont {Zhu}}, \bibinfo {author} {\bibfnamefont
  {G.-Y.}\ \bibnamefont {Xiang}}, \bibinfo {author} {\bibfnamefont {C.-F.}\
  \bibnamefont {Li}},\ and\ \bibinfo {author} {\bibfnamefont {G.-C.}\
  \bibnamefont {Guo}},\ }\bibfield  {title} {\bibinfo {title} {Experimental
  optimal orienteering via parallel and antiparallel spins},\ }\href@noop {}
  {\bibfield  {journal} {\bibinfo  {journal} {Phys. Rev. Lett.}\ }\textbf
  {\bibinfo {volume} {124}},\ \bibinfo {pages} {060502} (\bibinfo {year}
  {2020})}\BibitemShut {NoStop}%
\bibitem [{\citenamefont {Zhou}\ \emph {et~al.}(2025)\citenamefont {Zhou},
  \citenamefont {Yi}, \citenamefont {Yan}, \citenamefont {Hou}, \citenamefont
  {Zhu}, \citenamefont {Xiang}, \citenamefont {Li},\ and\ \citenamefont
  {Guo}}]{ZhouYYH23}%
  \BibitemOpen
  \bibfield  {author} {\bibinfo {author} {\bibfnamefont {K.}~\bibnamefont
  {Zhou}}, \bibinfo {author} {\bibfnamefont {C.}~\bibnamefont {Yi}}, \bibinfo
  {author} {\bibfnamefont {W.-Z.}\ \bibnamefont {Yan}}, \bibinfo {author}
  {\bibfnamefont {Z.}~\bibnamefont {Hou}}, \bibinfo {author} {\bibfnamefont
  {H.}~\bibnamefont {Zhu}}, \bibinfo {author} {\bibfnamefont {G.-Y.}\
  \bibnamefont {Xiang}}, \bibinfo {author} {\bibfnamefont {C.-F.}\ \bibnamefont
  {Li}},\ and\ \bibinfo {author} {\bibfnamefont {G.-C.}\ \bibnamefont {Guo}},\
  }\bibfield  {title} {\bibinfo {title} {Experimental realization of genuine
  three-copy collective measurements for optimal information extraction},\
  }\href@noop {} {\bibfield  {journal} {\bibinfo  {journal} {Phys. Rev. Lett.}\
  }\textbf {\bibinfo {volume} {134}},\ \bibinfo {pages} {210201} (\bibinfo
  {year} {2025})}\BibitemShut {NoStop}%
\bibitem [{\citenamefont {Conlon}\ \emph
  {et~al.}(2023{\natexlab{a}})\citenamefont {Conlon}, \citenamefont {Vogl},
  \citenamefont {Marciniak}, \citenamefont {Pogorelov}, \citenamefont {Yung},
  \citenamefont {Eilenberger}, \citenamefont {Berry}, \citenamefont {Santana},
  \citenamefont {Blatt}, \citenamefont {Monz} \emph {et~al.}}]{ConlVMP23}%
  \BibitemOpen
  \bibfield  {author} {\bibinfo {author} {\bibfnamefont {L.~O.}\ \bibnamefont
  {Conlon}}, \bibinfo {author} {\bibfnamefont {T.}~\bibnamefont {Vogl}},
  \bibinfo {author} {\bibfnamefont {C.~D.}\ \bibnamefont {Marciniak}}, \bibinfo
  {author} {\bibfnamefont {I.}~\bibnamefont {Pogorelov}}, \bibinfo {author}
  {\bibfnamefont {S.~K.}\ \bibnamefont {Yung}}, \bibinfo {author}
  {\bibfnamefont {F.}~\bibnamefont {Eilenberger}}, \bibinfo {author}
  {\bibfnamefont {D.~W.}\ \bibnamefont {Berry}}, \bibinfo {author}
  {\bibfnamefont {F.~S.}\ \bibnamefont {Santana}}, \bibinfo {author}
  {\bibfnamefont {R.}~\bibnamefont {Blatt}}, \bibinfo {author} {\bibfnamefont
  {T.}~\bibnamefont {Monz}}, \emph {et~al.},\ }\bibfield  {title} {\bibinfo
  {title} {Approaching optimal entangling collective measurements on quantum
  computing platforms},\ }\href@noop {} {\bibfield  {journal} {\bibinfo
  {journal} {Nat. Phys.}\ }\textbf {\bibinfo {volume} {19}},\ \bibinfo {pages}
  {351} (\bibinfo {year} {2023}{\natexlab{a}})}\BibitemShut {NoStop}%
\bibitem [{\citenamefont {Conlon}\ \emph
  {et~al.}(2023{\natexlab{b}})\citenamefont {Conlon}, \citenamefont
  {Eilenberger}, \citenamefont {Lam},\ and\ \citenamefont {Assad}}]{ConlELA23}%
  \BibitemOpen
  \bibfield  {author} {\bibinfo {author} {\bibfnamefont {L.~O.}\ \bibnamefont
  {Conlon}}, \bibinfo {author} {\bibfnamefont {F.}~\bibnamefont {Eilenberger}},
  \bibinfo {author} {\bibfnamefont {P.~K.}\ \bibnamefont {Lam}},\ and\ \bibinfo
  {author} {\bibfnamefont {S.~M.}\ \bibnamefont {Assad}},\ }\bibfield  {title}
  {\bibinfo {title} {Discriminating mixed qubit states with collective
  measurements},\ }\href@noop {} {\bibfield  {journal} {\bibinfo  {journal}
  {Commun. Phys.}\ }\textbf {\bibinfo {volume} {6}},\ \bibinfo {pages} {337}
  (\bibinfo {year} {2023}{\natexlab{b}})}\BibitemShut {NoStop}%
\bibitem [{\citenamefont {Tian}\ \emph {et~al.}(2024)\citenamefont {Tian},
  \citenamefont {Yan}, \citenamefont {Hou}, \citenamefont {Xiang},
  \citenamefont {Li},\ and\ \citenamefont {Guo}}]{TianYHX24}%
  \BibitemOpen
  \bibfield  {author} {\bibinfo {author} {\bibfnamefont {B.}~\bibnamefont
  {Tian}}, \bibinfo {author} {\bibfnamefont {W.-Z.}\ \bibnamefont {Yan}},
  \bibinfo {author} {\bibfnamefont {Z.}~\bibnamefont {Hou}}, \bibinfo {author}
  {\bibfnamefont {G.-Y.}\ \bibnamefont {Xiang}}, \bibinfo {author}
  {\bibfnamefont {C.-F.}\ \bibnamefont {Li}},\ and\ \bibinfo {author}
  {\bibfnamefont {G.-C.}\ \bibnamefont {Guo}},\ }\bibfield  {title} {\bibinfo
  {title} {Minimum-consumption discrimination of quantum states via globally
  optimal adaptive measurements},\ }\href@noop {} {\bibfield  {journal}
  {\bibinfo  {journal} {Phys. Rev. Lett.}\ }\textbf {\bibinfo {volume} {132}},\
  \bibinfo {pages} {110801} (\bibinfo {year} {2024})}\BibitemShut {NoStop}%
\bibitem [{\citenamefont {Horodecki}\ \emph {et~al.}(2009)\citenamefont
  {Horodecki}, \citenamefont {Horodecki}, \citenamefont {Horodecki},\ and\
  \citenamefont {Horodecki}}]{HoroHHH09}%
  \BibitemOpen
  \bibfield  {author} {\bibinfo {author} {\bibfnamefont {R.}~\bibnamefont
  {Horodecki}}, \bibinfo {author} {\bibfnamefont {P.}~\bibnamefont
  {Horodecki}}, \bibinfo {author} {\bibfnamefont {M.}~\bibnamefont
  {Horodecki}},\ and\ \bibinfo {author} {\bibfnamefont {K.}~\bibnamefont
  {Horodecki}},\ }\bibfield  {title} {\bibinfo {title} {Quantum entanglement},\
  }\href@noop {} {\bibfield  {journal} {\bibinfo  {journal} {Rev. Mod. Phys.}\
  }\textbf {\bibinfo {volume} {81}},\ \bibinfo {pages} {865} (\bibinfo {year}
  {2009})}\BibitemShut {NoStop}%
\bibitem [{\citenamefont {G\"uhne}\ and\ \citenamefont
  {T\'oth}(2009)}]{GuhnT09}%
  \BibitemOpen
  \bibfield  {author} {\bibinfo {author} {\bibfnamefont {O.}~\bibnamefont
  {G\"uhne}}\ and\ \bibinfo {author} {\bibfnamefont {G.}~\bibnamefont
  {T\'oth}},\ }\bibfield  {title} {\bibinfo {title} {Entanglement detection},\
  }\href@noop {} {\bibfield  {journal} {\bibinfo  {journal} {Phys. Rep.}\
  }\textbf {\bibinfo {volume} {474}},\ \bibinfo {pages} {1} (\bibinfo {year}
  {2009})}\BibitemShut {NoStop}%
\bibitem [{\citenamefont {Brunner}\ \emph {et~al.}(2014)\citenamefont
  {Brunner}, \citenamefont {Cavalcanti}, \citenamefont {Pironio}, \citenamefont
  {Scarani},\ and\ \citenamefont {Wehner}}]{BrunCPS13}%
  \BibitemOpen
  \bibfield  {author} {\bibinfo {author} {\bibfnamefont {N.}~\bibnamefont
  {Brunner}}, \bibinfo {author} {\bibfnamefont {D.}~\bibnamefont {Cavalcanti}},
  \bibinfo {author} {\bibfnamefont {S.}~\bibnamefont {Pironio}}, \bibinfo
  {author} {\bibfnamefont {V.}~\bibnamefont {Scarani}},\ and\ \bibinfo {author}
  {\bibfnamefont {S.}~\bibnamefont {Wehner}},\ }\bibfield  {title} {\bibinfo
  {title} {{Bell} nonlocality},\ }\href@noop {} {\bibfield  {journal} {\bibinfo
   {journal} {Rev. Mod. Phys.}\ }\textbf {\bibinfo {volume} {86}},\ \bibinfo
  {pages} {419} (\bibinfo {year} {2014})}\BibitemShut {NoStop}%
\bibitem [{\citenamefont {Derka}\ \emph {et~al.}(1998)\citenamefont {Derka},
  \citenamefont {Bu{\v{z}}ek},\ and\ \citenamefont {Ekert}}]{DerkBE98}%
  \BibitemOpen
  \bibfield  {author} {\bibinfo {author} {\bibfnamefont {R.}~\bibnamefont
  {Derka}}, \bibinfo {author} {\bibfnamefont {V.}~\bibnamefont {Bu{\v{z}}ek}},\
  and\ \bibinfo {author} {\bibfnamefont {A.~K.}\ \bibnamefont {Ekert}},\
  }\bibfield  {title} {\bibinfo {title} {Universal algorithm for optimal
  estimation of quantum states from finite ensembles via realizable generalized
  measurement},\ }\href@noop {} {\bibfield  {journal} {\bibinfo  {journal}
  {Phys. Rev. Lett.}\ }\textbf {\bibinfo {volume} {80}},\ \bibinfo {pages}
  {1571} (\bibinfo {year} {1998})}\BibitemShut {NoStop}%
\bibitem [{\citenamefont {Latorre}\ \emph {et~al.}(1998)\citenamefont
  {Latorre}, \citenamefont {Pascual},\ and\ \citenamefont
  {Tarrach}}]{LatoPT98}%
  \BibitemOpen
  \bibfield  {author} {\bibinfo {author} {\bibfnamefont {J.~I.}\ \bibnamefont
  {Latorre}}, \bibinfo {author} {\bibfnamefont {P.}~\bibnamefont {Pascual}},\
  and\ \bibinfo {author} {\bibfnamefont {R.}~\bibnamefont {Tarrach}},\
  }\bibfield  {title} {\bibinfo {title} {Minimal optimal generalized quantum
  measurements},\ }\href@noop {} {\bibfield  {journal} {\bibinfo  {journal}
  {Phys. Rev. Lett.}\ }\textbf {\bibinfo {volume} {81}},\ \bibinfo {pages}
  {1351} (\bibinfo {year} {1998})}\BibitemShut {NoStop}%
\bibitem [{\citenamefont {Hayashi}(1998)}]{Haya98}%
  \BibitemOpen
  \bibfield  {author} {\bibinfo {author} {\bibfnamefont {M.}~\bibnamefont
  {Hayashi}},\ }\bibfield  {title} {\bibinfo {title} {Asymptotic estimation
  theory for a finite-dimensional pure state model},\ }\href@noop {} {\bibfield
   {journal} {\bibinfo  {journal} {J. Phys. A: Math. Gen.}\ }\textbf {\bibinfo
  {volume} {31}},\ \bibinfo {pages} {4633} (\bibinfo {year}
  {1998})}\BibitemShut {NoStop}%
\bibitem [{\citenamefont {Bru\ss{}}\ and\ \citenamefont
  {Macchiavello}(1999)}]{BrusM99}%
  \BibitemOpen
  \bibfield  {author} {\bibinfo {author} {\bibfnamefont {D.}~\bibnamefont
  {Bru\ss{}}}\ and\ \bibinfo {author} {\bibfnamefont {C.}~\bibnamefont
  {Macchiavello}},\ }\bibfield  {title} {\bibinfo {title} {Optimal state
  estimation for $d$-dimensional quantum systems},\ }\href@noop {} {\bibfield
  {journal} {\bibinfo  {journal} {Phys. Lett. A}\ }\textbf {\bibinfo {volume}
  {253}},\ \bibinfo {pages} {249} (\bibinfo {year} {1999})}\BibitemShut
  {NoStop}%
\bibitem [{\citenamefont {Ac\'{\i}n}\ \emph {et~al.}(2000)\citenamefont
  {Ac\'{\i}n}, \citenamefont {Latorre},\ and\ \citenamefont
  {Pascual}}]{AcinLP00}%
  \BibitemOpen
  \bibfield  {author} {\bibinfo {author} {\bibfnamefont {A.}~\bibnamefont
  {Ac\'{\i}n}}, \bibinfo {author} {\bibfnamefont {J.~I.}\ \bibnamefont
  {Latorre}},\ and\ \bibinfo {author} {\bibfnamefont {P.}~\bibnamefont
  {Pascual}},\ }\bibfield  {title} {\bibinfo {title} {Optimal generalized
  quantum measurements for arbitrary spin systems},\ }\href@noop {} {\bibfield
  {journal} {\bibinfo  {journal} {Phys. Rev. A}\ }\textbf {\bibinfo {volume}
  {61}},\ \bibinfo {pages} {022113} (\bibinfo {year} {2000})}\BibitemShut
  {NoStop}%
\bibitem [{\citenamefont {Bagan}\ \emph {et~al.}(2002)\citenamefont {Bagan},
  \citenamefont {Baig},\ and\ \citenamefont {{Mu\~noz-Tapia}}}]{BagaBM02}%
  \BibitemOpen
  \bibfield  {author} {\bibinfo {author} {\bibfnamefont {E.}~\bibnamefont
  {Bagan}}, \bibinfo {author} {\bibfnamefont {M.}~\bibnamefont {Baig}},\ and\
  \bibinfo {author} {\bibfnamefont {R.}~\bibnamefont {{Mu\~noz-Tapia}}},\
  }\bibfield  {title} {\bibinfo {title} {Optimal scheme for estimating a pure
  qubit state via local measurements},\ }\href@noop {} {\bibfield  {journal}
  {\bibinfo  {journal} {Phys. Rev. Lett.}\ }\textbf {\bibinfo {volume} {89}},\
  \bibinfo {pages} {277904} (\bibinfo {year} {2002})}\BibitemShut {NoStop}%
\bibitem [{\citenamefont {Bagan}\ \emph {et~al.}(2005)\citenamefont {Bagan},
  \citenamefont {Monras},\ and\ \citenamefont {Mu\~noz Tapia}}]{BagaMM05}%
  \BibitemOpen
  \bibfield  {author} {\bibinfo {author} {\bibfnamefont {E.}~\bibnamefont
  {Bagan}}, \bibinfo {author} {\bibfnamefont {A.}~\bibnamefont {Monras}},\ and\
  \bibinfo {author} {\bibfnamefont {R.}~\bibnamefont {Mu\~noz Tapia}},\
  }\bibfield  {title} {\bibinfo {title} {Comprehensive analysis of quantum
  pure-state estimation for two-level systems},\ }\href@noop {} {\bibfield
  {journal} {\bibinfo  {journal} {Phys. Rev. A}\ }\textbf {\bibinfo {volume}
  {71}},\ \bibinfo {pages} {062318} (\bibinfo {year} {2005})}\BibitemShut
  {NoStop}%
\bibitem [{\citenamefont {Hayashi}\ \emph {et~al.}(2005)\citenamefont
  {Hayashi}, \citenamefont {Hashimoto},\ and\ \citenamefont
  {Horibe}}]{HayaHH05}%
  \BibitemOpen
  \bibfield  {author} {\bibinfo {author} {\bibfnamefont {A.}~\bibnamefont
  {Hayashi}}, \bibinfo {author} {\bibfnamefont {T.}~\bibnamefont {Hashimoto}},\
  and\ \bibinfo {author} {\bibfnamefont {M.}~\bibnamefont {Horibe}},\
  }\bibfield  {title} {\bibinfo {title} {Reexamination of optimal quantum state
  estimation of pure states},\ }\href@noop {} {\bibfield  {journal} {\bibinfo
  {journal} {Phys. Rev. A}\ }\textbf {\bibinfo {volume} {72}},\ \bibinfo
  {pages} {032325} (\bibinfo {year} {2005})}\BibitemShut {NoStop}%
\bibitem [{\citenamefont {Martens}\ and\ \citenamefont
  {de~Muynck}(1990)}]{MartM90}%
  \BibitemOpen
  \bibfield  {author} {\bibinfo {author} {\bibfnamefont {H.}~\bibnamefont
  {Martens}}\ and\ \bibinfo {author} {\bibfnamefont {W.~M.}\ \bibnamefont
  {de~Muynck}},\ }\bibfield  {title} {\bibinfo {title} {Nonideal quantum
  measurements},\ }\href@noop {} {\bibfield  {journal} {\bibinfo  {journal}
  {Found. Phys.}\ }\textbf {\bibinfo {volume} {20}},\ \bibinfo {pages} {255}
  (\bibinfo {year} {1990})}\BibitemShut {NoStop}%
\bibitem [{\citenamefont {Zhu}(2022)}]{Zhu22}%
  \BibitemOpen
  \bibfield  {author} {\bibinfo {author} {\bibfnamefont {H.}~\bibnamefont
  {Zhu}},\ }\bibfield  {title} {\bibinfo {title} {Quantum measurements in the
  light of quantum state estimation},\ }\href@noop {} {\bibfield  {journal}
  {\bibinfo  {journal} {PRX Quantum}\ }\textbf {\bibinfo {volume} {3}},\
  \bibinfo {pages} {030306} (\bibinfo {year} {2022})}\BibitemShut {NoStop}%
\bibitem [{\citenamefont {Ivanovi\'c}(1981)}]{Ivan81}%
  \BibitemOpen
  \bibfield  {author} {\bibinfo {author} {\bibfnamefont {I.~D.}\ \bibnamefont
  {Ivanovi\'c}},\ }\bibfield  {title} {\bibinfo {title} {Geometrical
  description of quantal state determination},\ }\href@noop {} {\bibfield
  {journal} {\bibinfo  {journal} {J. Phys. A: Math. Gen.}\ }\textbf {\bibinfo
  {volume} {14}},\ \bibinfo {pages} {3241} (\bibinfo {year}
  {1981})}\BibitemShut {NoStop}%
\bibitem [{\citenamefont {Wootters}\ and\ \citenamefont
  {Fields}(1989)}]{WootF89}%
  \BibitemOpen
  \bibfield  {author} {\bibinfo {author} {\bibfnamefont {W.~K.}\ \bibnamefont
  {Wootters}}\ and\ \bibinfo {author} {\bibfnamefont {B.~D.}\ \bibnamefont
  {Fields}},\ }\bibfield  {title} {\bibinfo {title} {Optimal
  state--determination by mutually unbiased measurements},\ }\href@noop {}
  {\bibfield  {journal} {\bibinfo  {journal} {Ann. Phys.}\ }\textbf {\bibinfo
  {volume} {191}},\ \bibinfo {pages} {363} (\bibinfo {year}
  {1989})}\BibitemShut {NoStop}%
\bibitem [{\citenamefont {Durt}\ \emph {et~al.}(2010)\citenamefont {Durt},
  \citenamefont {Englert}, \citenamefont {Bengtsson},\ and\ \citenamefont
  {{\.{Z}}yczkowski}}]{DurtEBZ10}%
  \BibitemOpen
  \bibfield  {author} {\bibinfo {author} {\bibfnamefont {T.}~\bibnamefont
  {Durt}}, \bibinfo {author} {\bibfnamefont {B.-G.}\ \bibnamefont {Englert}},
  \bibinfo {author} {\bibfnamefont {I.}~\bibnamefont {Bengtsson}},\ and\
  \bibinfo {author} {\bibfnamefont {K.}~\bibnamefont {{\.{Z}}yczkowski}},\
  }\bibfield  {title} {\bibinfo {title} {On mutually unbiased bases},\
  }\href@noop {} {\bibfield  {journal} {\bibinfo  {journal} {Int. J. Quantum
  Inf.}\ }\textbf {\bibinfo {volume} {08}},\ \bibinfo {pages} {535} (\bibinfo
  {year} {2010})}\BibitemShut {NoStop}%
\bibitem [{\citenamefont {Dankert}\ \emph {et~al.}(2009)\citenamefont
  {Dankert}, \citenamefont {Cleve}, \citenamefont {Emerson},\ and\
  \citenamefont {Livine}}]{DankCEL09}%
  \BibitemOpen
  \bibfield  {author} {\bibinfo {author} {\bibfnamefont {C.}~\bibnamefont
  {Dankert}}, \bibinfo {author} {\bibfnamefont {R.}~\bibnamefont {Cleve}},
  \bibinfo {author} {\bibfnamefont {J.}~\bibnamefont {Emerson}},\ and\ \bibinfo
  {author} {\bibfnamefont {E.}~\bibnamefont {Livine}},\ }\bibfield  {title}
  {\bibinfo {title} {Exact and approximate unitary 2-designs and their
  application to fidelity estimation},\ }\href@noop {} {\bibfield  {journal}
  {\bibinfo  {journal} {Phys. Rev. A}\ }\textbf {\bibinfo {volume} {80}},\
  \bibinfo {pages} {012304} (\bibinfo {year} {2009})}\BibitemShut {NoStop}%
\bibitem [{\citenamefont {Gross}\ \emph {et~al.}(2007)\citenamefont {Gross},
  \citenamefont {Audenaert},\ and\ \citenamefont {Eisert}}]{GrosAE07}%
  \BibitemOpen
  \bibfield  {author} {\bibinfo {author} {\bibfnamefont {D.}~\bibnamefont
  {Gross}}, \bibinfo {author} {\bibfnamefont {K.}~\bibnamefont {Audenaert}},\
  and\ \bibinfo {author} {\bibfnamefont {J.}~\bibnamefont {Eisert}},\
  }\bibfield  {title} {\bibinfo {title} {Evenly distributed unitaries: On the
  structure of unitary designs},\ }\href@noop {} {\bibfield  {journal}
  {\bibinfo  {journal} {J. Math. Phys.}\ }\textbf {\bibinfo {volume} {48}},\
  \bibinfo {pages} {052104} (\bibinfo {year} {2007})}\BibitemShut {NoStop}%
\bibitem [{\citenamefont {Zhu}(2017)}]{Zhu17MC}%
  \BibitemOpen
  \bibfield  {author} {\bibinfo {author} {\bibfnamefont {H.}~\bibnamefont
  {Zhu}},\ }\bibfield  {title} {\bibinfo {title} {Multiqubit {Clifford} groups
  are unitary 3-designs},\ }\href@noop {} {\bibfield  {journal} {\bibinfo
  {journal} {Phys. Rev. A}\ }\textbf {\bibinfo {volume} {96}},\ \bibinfo
  {pages} {062336} (\bibinfo {year} {2017})}\BibitemShut {NoStop}%
\bibitem [{\citenamefont {Webb}(2016)}]{Webb16}%
  \BibitemOpen
  \bibfield  {author} {\bibinfo {author} {\bibfnamefont {Z.}~\bibnamefont
  {Webb}},\ }\bibfield  {title} {\bibinfo {title} {The {C}lifford group forms a
  unitary 3-design},\ }\href@noop {} {\bibfield  {journal} {\bibinfo  {journal}
  {Quantum Inf. Comput.}\ }\textbf {\bibinfo {volume} {16}},\ \bibinfo {pages}
  {1379–1400} (\bibinfo {year} {2016})}\BibitemShut {NoStop}%
\bibitem [{\citenamefont {Hill}\ and\ \citenamefont
  {Wootters}(1997)}]{HillW97}%
  \BibitemOpen
  \bibfield  {author} {\bibinfo {author} {\bibfnamefont {S.}~\bibnamefont
  {Hill}}\ and\ \bibinfo {author} {\bibfnamefont {W.~K.}\ \bibnamefont
  {Wootters}},\ }\bibfield  {title} {\bibinfo {title} {Entanglement of a pair
  of quantum bits},\ }\href@noop {} {\bibfield  {journal} {\bibinfo  {journal}
  {Phys. Rev. Lett.}\ }\textbf {\bibinfo {volume} {78}},\ \bibinfo {pages}
  {5022} (\bibinfo {year} {1997})}\BibitemShut {NoStop}%
\end{thebibliography}%

\end{document}